\documentclass[10pt]{article}
\usepackage{amsmath, amsfonts, amsthm, amssymb,verbatim}
        \newtheorem{theorem}{Theorem}[section]
\newtheorem{definition}[theorem]{Definition}
        \newtheorem{proposition}[theorem]{Proposition}
        \newtheorem{lemma}[theorem]{Lemma}

\numberwithin{equation}{section}

\newcommand \II {\text{II}}
\newcommand \xbf {{\mbox{\boldmath$x$}}}
\newcommand \ybf {{\mbox{\boldmath$y$}}}

\newcommand \etabf {\mbox{\boldmath$\eta$}}
\newcommand \Gammabf {\mbox{\boldmath$\Gamma$}}

\newcommand \nablat         {\mbox{\boldmath$\widetilde \nabla$}} 

\newcommand \Mbf {\mathbf M}
\newcommand \pbf {\mathbf p}

\newcommand \qbf {\mathbf q}

\newcommand \Nbf {\mathbf N}

\newcommand \St {{\widetilde S}}

\newcommand \la \langle
\newcommand \ra \rangle
\newcommand \tbar {{\overline t}}
\newcommand \mbar {\overline m}
\newcommand \Cbar {\overline C}
\newcommand \bart {\underline t}
\newcommand \barm {\underline m}
\newcommand \taubar {\overline \tau}
\newcommand \bartau {\underline \tau}

\newcommand \baru {\underline u}
\newcommand \barc {\underline c}
\newcommand \cbar {\overline c}
\newcommand \Acal {\mathcal A}
\newcommand \Lcal {\mathcal L}

\newcommand \bark {\underline k}
\newcommand \Rmax {\mathbf R_{\text{\bf max}}}

\newcommand \Ricb       {\text{\bf Ric}}
\newcommand \gbf        {\mathbf g}
\newcommand \gbfhat     {\widehat {\mathbf g}}
\newcommand \Tbf        {\mathbf T}

\newcommand \Dbf      {\mbox{\boldmath$\nabla$}}
\newcommand \Dbfhat      {\mbox{\boldmath$\widehat\nabla$}}

\newcommand \Rbf        {\mathbf R}

\newcommand \be     {\begin{equation}}
\newcommand \ee     {\end{equation}}

\newcommand \kbar   {\overline k}
\newcommand \del        \partial
\newcommand \eps     \epsilon
\newcommand \auth   \textsc

\newcommand \Mcal    {{\mathcal M}}
\newcommand \Ccal    {{\mathcal C}}

\newcommand \Hcal   {{\mathcal H}}
\newcommand \Bcal   {{\mathcal B}}

\newcommand \gt     {{\widetilde {\mathbf g}}}
\newcommand \dt     {{\widetilde {\mathbf d}}}

\newcommand \inj        {\text {\bf Inj}}

\newcommand \Riem       {\text{\bf Rm}}

\newcommand \Rm       {\text{\bf Rm}}
\newcommand \Ric       {\text{Ric}}
\newcommand \Ricbf       {\text{\bf Ric}}
\newcommand \Gammat  {\widetilde{\Gamma}}

\newcommand \expb   {{{\text{\bf exp}}}}

\begin{document}

\title{Local foliations and optimal regularity of Einstein spacetimes}
\author{Bing-Long Chen$^1$
and
Philippe G. LeFloch$^2$}

\date{October 2, 2008}

\maketitle

\footnotetext[1]{Department of Mathematics, Sun Yat-Sen University, Guangzhou, People's Republic of China. 
E-mail: {\sl mcscbl@mail.sysu.edu.cn.}
}

\footnotetext[2]{Laboratoire Jacques-Louis Lions \& Centre National de la Recherche Scientifique,
Universit\'e Pierre et Marie Curie (Paris 6), 4 Place Jussieu,  75252 Paris, France.
E-mail: {\sl LeFloch@ann.jussieu.fr.} 
}

\begin{abstract}
We investigate the local regularity of pointed spacetimes, that is, time-oriented Lorent\-zian
manifolds in which a point and a future-oriented, unit timelike vector (an observer) are selected.
Our main result covers the class of Einstein vacuum spacetimes.
Under curvature and injectivity bounds only,
we establish the existence of a {\sl local coordinate} chart defined in a ball with definite size
in which the metric coefficients have optimal regularity.
The proof is based on {\sl quantitative estimates}, on one hand, for 
a constant mean curvature (CMC) foliation 
by spacelike hypersurfaces defined locally near the observer and, on the other hand, 
for the metric in local coordinates that are spatially harmonic in each CMC slice.
The results and techniques in this paper should be useful in the context of general relativity
for investigating the long-time behavior of solutions to the Einstein equations.
\end{abstract}
 

\section{Introduction}
\label{IN-0}

\subsection{Quantitative estimates for CMC foliations}

We denote by $(\Mbf,\gbf)$ a spacetime of general relativity, that is,
a time-oriented, $(n+1)$-dimensional Lorentzian manifold
whose metric $\gbf$, by definition, has signature $(-, +, \ldots, +)$. 
Our main result in the present paper will concern vacuum spacetimes, that is, Ricci-flat manifolds, 
although this assumption will be made only later in the discussion. 
Building on our earlier work \cite{ChenLeFloch}, we continue the investigation of the local geometry
of Einstein spacetimes, using here techniques for partial differential equations. 
Our main objective will be, under natural geometric bounds on the curvature and the injectivity radius
only,
 to establish the existence of
local coordinate charts in which the metric coefficients have optimal regularity, that is,
belong to the Sobolev space $W^{2,a}$ for all real $a \in (1,\infty)$.
The construction proposed in the present paper is local in the neighborhood of 
a given ``observer''
and, in turn,
 our result provides a sharp control of the local geometry of the spacetime at every point.
This optimal regularity theory should be useful for tackling the global regularity issue for Einstein
spacetimes and investigating the long-time behavior of solutions to the Einstein equations.

As in \cite{ChenLeFloch}, we consider a {\sl pointed} Lorentzian manifold $(\Mbf, \gbf, \pbf,\Tbf_\pbf)$,
that is, a time-oriented Lorentzian manifold supplemented with a point $\pbf \in \Mbf$
and a future-oriented timelike vector $\Tbf_\pbf$ at that point. The pair $(\pbf, \Tbf_\pbf)$ is called a
{\sl (local) observer} and is required for stating our curvature and injectivity radius bounds
for some given curvature constant $\Lambda$ and injectivity radius constant $\lambda>0$; see~\eqref{bd3} in
Section~\ref{BB-01} below.  
We will establish the existence of a
neighborhood of the observer $(\pbf,\Tbf_\pbf)$ whose size depends
on $\Lambda, \lambda$ only and in which local coordinates exist
in such a way that the regularity of the metric coefficients can be controlled by the
same constants. Since the Riemann curvature involves up to two
derivatives of the metric it is natural to search for an estimate of the metric in the $W^{2,a}$ norm, 
and this is precisely what we achieve in the present paper. 

We will proceed as follows. Our first task is constructing a {\sl constant mean curvature} (CMC) foliation
by spacelike hypersurfaces,
which is locally defined near the observer and satisfies quantitative bounds involving the constants $\Lambda, \lambda$, only;
see Theorem~\ref{foli} below. Our method can be viewed as a refinement of earlier works by
Bartnik and Simon \cite{BartnikSimon} (covering hypersurfaces in Minkowski space) and Gerhardt \cite{Gerhardt,Gerhardt1}
(global foliations of Lorentzian manifolds).  If one would assume that the metric $\gbf$
admits bounded covariant derivatives of sufficiently high order of regularity, then the
techniques in \cite{BartnikSimon,Gerhardt} would provide the existence of the CMC foliation
and certain estimates. Hence, 
the construction of a CMC foliation on a {\sl sufficiently smooth} manifold is standard at {\sl small scales.}

In contrast, in the framework of the present paper only {\sl
limited differentiability} of the metric should be used and uniform bounds involving the curvature and injectivity 
radius bounds, only, be sought. We have
to solve a boundary value problem for the prescribed mean
curvature equation in a Lorentzian background and to establish
that a CMC foliation exists in a neighborhood (of the observer)
with {\sl definite size} and to control the geometry of these
slices in terms of $\Lambda, \lambda$, only. A technical
difficulty in this analysis is ensuring that each hypersurface of
the foliation is {\sl uniformly} spacelike and can not approach a
null hypersurface. Deriving a gradient estimate for prescribed
curvature equations requires the use of barrier functions
determined from (parts of) suitably constructed geodesic spheres.

\subsection{Earlier works}

An extensive study of (sufficiently regular) spacetimes
admitting {\sl global} foliations by spatially compact hypersurfaces with constant mean curvature
is available in the literature.  
Andersson and Moncrief \cite{AM1,AM2} and Andersson \cite{Andersson0,Andersson} 
have established global existence theorems 
for sufficiently small perturbations of a large family of spacetimes. 
For instance, their method allowed them to establish 
a global existence theorem for sufficiently small perturbation of Friedmann-Robertson-Walker type
spacetimes. Their construction is based 
on constructing a global CMC foliation and uses harmonic coordinates on each slice. 
In these works, the authors derive (and strongly rely on) a~priori estimates which are 
based on the so-called Bel-Robinson tensor and involve {\sl up to third-order} derivatives of the metric. 
In contrast, we focus in the present paper on the local existence of such foliations 
but require only the sup norm of the curvature to be bounded. 
Our new approach leads to a construction of ``good'' local coordinates (see below) 
and allows us to explore the local optimal regularity of Lorentzian metrics. 
Another direction of research on CMC foliations 
is currently developed by Reiris \cite{Reiris,Reiris2}, who analyzes the CMC Einstein flow in connection with the Bel-Robinson energy and also imposes higher regularity of the metric. 

We also refer the reader to an ambitious program (the $L^2$ curvature conjecture)
initiated and developed by Klainerman and Rodnianski in a series of papers; see
\cite{KR1,KR2,KR4,KR5}. In these works, the authors 
are interested in controling the geometry of {\sl null cones}
 which may become singular due to caustic formation. 
The regularity of null cones is needed in order to suitably extend the methods
of harmonic analysis to the Einstein equations. 
In particular, the recent result \cite{KR5} provides a breakdown criterion for solutions to the Einstein equations. 
In comparison with the present work, 
the objectives in \cite{KR5} are different: 
these authors rely on {\sl hyperbolic} techniques and investigate the geometry of {\sl light cones,}  
while our approach in the present paper is purely {\sl elliptic} in nature and addresses the
geometry of {\sl the spacetime itself.}  


\subsection{CMC--harmonic coordinates of an observer}

Our second task is constructing local coordinates.
In Riemannian geometry it is well-known that geodesic-based coordinates
and distance-based coordinates fail to achieve the optimal regularity of the metric.
The use of {\sl harmonic} coordinates on Riemannian manifolds was first advocated by De~Turck and Kazdan
\cite{DeTurckKazdan} and, later,
a quantitative bound on the harmonic radius at a point was derived by Jost and Karcher \cite{JostKarcher}
in terms of curvature and volume bounds, only.
More recently, the issue of the optimal regularity of Lorentzian metrics
was tackled by Anderson in the pioneering work \cite{Anderson-long-time,Anderson-regularity}. He proposed
to use a combination of normal coordinates (based on geodesics) and spatially harmonic coordinates,
and derived several uniform estimates for the metric coefficients. This construction based on 
geodesics does not lead to the desired optimal regularity, however. We also refer to earlier work by Anderson \cite{Anderson1,Anderson2} for further regularity results 
within the class of static and, more generally, stationary spacetimes. 

Our main result covers Einstein vacuum spacetimes, that is, manifolds satisfying the Ricci-flat condition
\be
\label{bd2}
\Ricbf_\gbf = 0,
\ee
and the construction we propose is as follows.
Relying on our quantitative estimates for CMC foliations near a given observer (Theorem~\ref{foli})
and then applying Jost and Karcher's theorem for Riemannian manifolds \cite{JostKarcher},
we construct (spatially) harmonic coordinates on each spacelike CMC slice. We refer to such coordinates
as {\sl CMC--harmonic coordinates,} and we prove first that, on every slice, the spatial metric
coefficients $\gbf_{ij}$ belong
to the Sobolev space $W^{2,a}$ and satisfy the {\sl quantitative} estimate
$$
\| \gbf_{ij} \|_{\pbf, \Tbf_\pbf, W^{2,a}} \leq C(a, \Lambda, \lambda)
$$
for all $a<\infty$ and some constant ${C(a, \Lambda, \lambda)>0}$ (depending also on the dimension $n$).
In addition, we also control the lapse function and the shift vector associated with
these local coordinates. The shift vector, denoted below by $\xi$, arises since coordinates are not simply transported from
one slice to another but are chosen to be harmonic on each slice. The
lapse function, denoted below by $\lambda$, is a measure of the distance between two nearby slices.

In turn, we arrive at the following main result of the present paper.

\begin{theorem}[CMC--harmonic coordinates of an observer]
\label{main}
There exist constants $0 < \barc(n) < c(n) < 1$ and $C(n), C_q(n) >0$ depending upon the dimension $n$
(and some exponent $q \in [1,\infty))$ such that the following properties hold.
Let $(\Mbf,\gbf, \pbf,\Tbf_\pbf)$ be an $(n+1)$-dimensional, pointed, Einstein vacuum spacetime satisfying
the following curvature and injectivity radius bounds at the scale $r>0$:
\be
\label{bd3}
\Rmax^r(\Mbf, \gbf, \pbf, \Tbf_\pbf) \leq r^{-2},
\qquad
\inj(\Mbf, \gbf, \pbf, \Tbf_\pbf) \geq r.
\ee
Then, there exists a local coordinate system $\xbf=(t,x^1,\ldots,x^n)$ having
$p={(r_1,0,\ldots,0)}$ for some $r_1\in [\barc(n) r, c(n) r]$
and defined for all
$$
|t-r_1| < c(n)^2 r, \qquad \big( (x^1)^2 + \ldots + (x^n)^2 \big)^{1/2} < c(n)^2 r,
$$
so that the following two properties hold:
\begin{itemize}
\item [i)] Each slice $\Sigma_t = \big\{(x^1)^2 + \ldots + (x^n)^2 < c(n)^4 r^2\big\}$ on which $t$ remains constant
is a spacelike hypersurface with constant mean curvature $c(n)^{-1} r^{-2} t$ and the coordinates
$x:=(x^1,\ldots,x^n)$ are harmonic for the metric induced on $\Sigma_t$.

\item[ii)] The Lorentzian metric in the spacetime coordinates $\xbf=(t,x^1,\ldots,x^n)$ has the form
\be
\label{metricform0}
\gbf = - \lambda(\xbf)^2 \, (dt)^2 + g_{ij} (\xbf) \big( dx^i + \xi^i(\xbf) \, dt \big) \big( dx^j + \xi^j(\xbf) \, dt \big)
\ee
and is close to the Minkowski metric in these local coordinates, in the sense that
$$
\aligned
& e^{-C(n)} \leq \lambda \leq e^{C(n)},
\\
& e^{-C(n)}\delta_{ij}\leq g_{ij} \leq  e^{C(n)}\delta_{ij},
\qquad
|\xi|_g^2 := g_{ij}\xi^i \xi^j \leq e^{-C(n)},
\endaligned
$$
and for each $q \in [1,\infty)$
$$
\frac{1}{r^{n-q}} \int_{\Sigma_t}|\del_{\mbox{\small \boldmath$x$}} \gbf|^q \, dv_{\Sigma_t}
+
\frac{1}{r^{n-2q}}\int_{\Sigma_t}|\del^2_{\mbox{\small \boldmath$x$} \mbox{\small \boldmath$x$}} \gbf|^q \, dv_{\Sigma_t}
\leq C_q(n).
$$
\end{itemize}
\end{theorem}

The theorem above establishes the existence of {\sl locally defined}
CMC--harmonic coordinates near any observer. The coordinates cover a neighborhood of the base point, whose
size is of order $r$ in the timelike and in the spacelike directions.
In the statement above, $\del_{\mbox{\small \boldmath$x$}} \gbf$ and
$\del_{\mbox{\small \boldmath$x$}\mbox{\small \boldmath$x$}} \gbf$
denote any {\sl spacetime} first- and second-order
derivatives of the metric coefficients in the local coordinates, respectively,
while $dv_{\Sigma_t}$ denotes the volume form induced on $\Sigma_t$ by the spacetime metric and can be
computed in terms of the spatial coordinates $x$.

Finally, let us put our results in a larger perspective.
The proposed framework relies on constructing purely {\sl local} CMC--harmonic coordinates
and, therefore, applies to spacetimes which need not admit a global CMC foliation.
Hwever, based on our local regularity theory we can also control the
{\sl global} geometry of the spacetime,
as follows.

Given a pointed Lorentzian manifold $(\Mbf, \gbf, \pbf, \Tbf_\pbf)$
satisfying a global version of the curvature and injectivity radius estimates
\eqref{bd3}, we can find a {\sl global} atlas of local charts
covering the whole of $\Mbf$ and in which the metric coefficients have the optimal regularity.
Such a conclusion is achieved by introducing a notion of {\sl global
CMC--harmonic radius viewed by the observer} $(\pbf, \Tbf_\pbf)$:
it is the ``largest'' radius $r>0$ such that {\sl local} CMC--harmonic coordinates exist
in a ball of radius $r$ about {\sl each} point
(and satisfy the uniform estimates stated in Theorem~\ref{main} for some fixed constants
$\barc(n), c(n), C(n), C_q(n)$). By establishing a lower bound on the radius of 
balls in which local CMC--harmonic coordinates exist at every point, we obtain
the desired global optimal regularity.
Again, this is a purely geometric result that involves the curvature and injectivity radius bounds, only.
This development is a work in progress.

Throughout this paper, we use the notation $C, C', C_1, \ldots$ for constants that only depend on the dimension
$n$ and may change at each occurrence. 


\section{CMC foliation of an observer}
\label{BB-0}

\subsection{Main statement in this section}
\label{BB-01}

In this section, we derive quantitative bounds on local CMC foliations for a general class of
Lorentzian manifolds which need not satisfy the Einstein equations.
For background on Riemannian or Lorentzian geometry we refer to \cite{Besse,HawkingEllis,ONeil}.
Let $(\Mbf,\gbf)$ be a time-oriented, $(n+1)$-dimensional Lorentzian manifold, and let
$\Dbf$ be the Levi-Civita connection associated with $\gbf$. Given a point $\pbf \in \Mbf$,
we want to construct a constant mean curvature foliation that is defined near $\pbf$
and whose geometry is uniformly controled in terms of the curvature and injectivity radius, only.
The inner product of two vectors $X,Y$ is also written $\la X, Y \ra_g = \la X,Y \ra = g(X,Y)$.

In fact, rather than a single point on the manifold we must prescribe
an {\sl observer,} that is, a pair $(\pbf,\Tbf_\pbf)$ where $\Tbf_\pbf$ is a unit, future-oriented, timelike vector at $\pbf$
(also called a {\sl reference vector}).
We use the notation $(\pbf,\Tbf_\pbf) \in T_1^+\Mbf$ for the bundle of such pairs, and we refer to
$(\Mbf, \gbf, \pbf,\Tbf_\pbf)$ as a {\sl pointed Lorentzian manifold.}
The vector $\Tbf_\pbf$ naturally induces a (positive-definite) inner product on the tangent space at $\pbf$,
which we denote by
$
\gbf_{\Tbf_\pbf} = \la \, \cdot \, , \, \cdot \, \ra_{\Tbf_\pbf}.
$
We sometimes write $|X|_{g_T}$ 
for the Riemannian norm of a vector $X$.
To simplify the notation, we often write $\Tbf$ instead of $\Tbf_\pbf$.

On a Lorentzian manifold the notion of injectivity radius is defined as follows.
Consider the {\sl exponential map} $\expb_\pbf$ at the point $\pbf$, as a map defined on
the Riemannian ball $B_{\gbf_\Tbf}(\pbf, r)$ (a subset of the tangent space at $\pbf$)
and taking values in $\Mbf$; this map is well-defined for all sufficiently small radius $r$, at least.

\begin{definition}
\label{radius}
Given a Lorentzian manifold $(\Mbf, \gbf)$, the {\sl injectivity radius} $\inj(\Mbf,\gbf,\pbf,\Tbf_\pbf)$
of an observer $(\pbf,\Tbf_\pbf) \in T_1^+\Mbf$ is the supremum over all
 radii $r>0$ such that the exponential map $\expb_\pbf$
is well-defined and is a global diffeomorphism from the subset $B_{\gbf_\Tbf}(\pbf,r)$ of the tangent space at $\pbf$
to a neighborhood of $\pbf$ in the manifold denoted by
$
\Bcal_{\gbf_\Tbf}(\pbf,r) := \expb_\pbf\big( B_{\gbf_\Tbf}(\pbf,r)\big) \subset \Mbf.
$
\end{definition}

To simplify the notation, we will also use the notation $B_\Tbf(\pbf,r)$ and $\Bcal_\Tbf(\pbf,r)$ for the
above Riemannian balls.
To state our assumption on the curvature we need a Riemaniann metric defined in a neighborhood of the point $\pbf$.
This reference metric is also denoted by $\gbf_\Tbf$ and is defined as follows.

By parallel transporting the vector $\Tbf_\pbf$,
with respect to the Lorentzian connection $\Dbf$ and
along radial geodesics leaving from $\pbf$, we construct a vector field $\Tbf$ which, however,  may be
 multi-valued
since two distinct geodesics leaving from $\pbf$, in general, may eventually intersect.
We use the notation $\Tbf_\gamma$ for the vector field defined along a radial geodesic $\gamma$
lying in the set $\Bcal_\Tbf(\pbf, r)$.
Then, to this vector field we canonically associate a positive-definite, inner product
$\gbf_{\Tbf_\gamma} = \la \, \cdot \, , \, \cdot \, \ra_{\Tbf_\gamma}$
defined in the tangent space of each point along the geodesic. We write
$|A|_{\Tbf_\gamma}$ or $|A|_\Tbf$ for the corresponding Riemannian norm of a tensor $A$.

We then consider the Riemann curvature $\Rm$ of the connection $\Dbf$
and, given an observer $(\pbf, \Tbf_\pbf)$, we compute its norm in the ball of radius $r$ 
\be
\label{maxcurv}
\Rmax^r(\Mbf,\gbf, \pbf,\Tbf_\pbf) := \sup_{\gamma} |\Riem|_{\Tbf_\gamma},
\ee
where the supremum is taken over every radial geodesic from $\pbf$ 
of Riemannian length $r$, at most. Note that
$|\Riem|_{\Tbf_\gamma}$ is evaluated with the Riemannian reference metric rather than from the Lorentzian metric.
Our main assumption \eqref{bd3} is now well-defined.

Our objective in this section is constructing a foliation near $\pbf$, say 
$
\bigcup_{\bart \leq t \leq \tbar} \Sigma_t,
$
by $n$-dimensional spacelike hypersurfaces $\Sigma_t \subset M$, and we require that each slice has constant mean curvature
equal to $t$, the range of $t$ being specified by some functions $\bart=\bart(\pbf)$ and $\tbar=\tbar(\pbf)$.
Moreover, this foliation should cover a ``relatively large'' part of the ball $\Bcal_\Tbf(\pbf, r)$.

\begin{theorem}[Uniform estimates for a local CMC foliation of an observer]
\label{foli}
There exist constants $\barc, c, \cbar, \theta,\zeta \in (0,1)$ with
$\barc < c < \cbar$, depending only on the dimension of the manifold
such that the following property holds. Let  $(\Mbf, \gbf,\pbf,\Tbf_\pbf)$ be a pointed Lorentzian
manifold satisfying the curvature and injectivity radius assumptions \eqref{bd3} at some scale $r>0$.
Then, the Riemannian ball $\Bcal_\Tbf(\pbf,cr)$ can be covered by a foliation
by spacelike hypersurfaces $\Sigma_t$ with constant mean curvature $t$,
\be
\label{foliation1}
\aligned
& \Big( \bigcup_{\bart \leq t \leq \tbar} \Sigma_t \Big) \supset \Bcal_\Tbf(\pbf,cr),
\qquad \quad
\bart := n {1-\zeta \over sr}, \qquad \tbar := n {1+\zeta \over sr},
\endaligned
\ee
in which the time variable describes a range $[\bart, \tbar]$ determined by some real $s \in [\barc,\cbar]$
and, moreover,
the unit normal vector $\Nbf$ and the second fundamental form $h$ of the foliation satisfy
$$
\aligned
1 \leq - g(\Nbf,\Tbf) & \leq 1 + \theta^{-1}, \qquad  \theta \leq -r^{-4} \gbf(\Dbf t, \Dbf t) \leq \theta^{-1}.
\qquad 
r \, |h|  \leq \theta^{-1},
\endaligned
$$
(Recall that the vector field $\Tbf$ is defined by parallel translating the given vector $\Tbf_\pbf$
along radial geodesics from $\pbf$.)
\end{theorem}

Hence, a foliation exists in a neighborhood of the base point, in which the time variable is
of order $1/r$ and describes an interval with definite size. In our construction given below, it 
will be important that
$s$ be chosen to be sufficiently small.

Note that the above theorem is purely geometric and does not
depend explicitly on the coordinates that we are going now to introduce in order to establish the existence
of the above foliation and control its geometry.

The rest of this section is devoted to giving a proof of Theorem~\ref{foli}. We first construct the
CMC hypersurfaces as graphs over geodesic spheres associated with the Lorentzian metric. Geodesic spheres
associated with the reference Riemannian metric will be introduced to serve as barrier functions.
Indeed, each CMC hypersurface will be pinched between a Lorentzian and a Riemannian geodesic ball.
The level set function describing the CMC hypersurface satisfies a nonlinear elliptic equation, whose
coefficients have rather limited regularity, and this will force us to use the Nash-Moser iteration technique.


\subsection{Formulation in normal coordinates}

\subsubsection*{Foliation by geodesic spheres}

We begin by introducing spacetime normal coordinates and by expressing the prescribed mean curvature equation
in these coordinates. As we established earlier in \cite{ChenLeFloch}, under the curvature and injectivity radius assumption
\eqref{bd3} for the observer $(\pbf,\Tbf_\pbf)$, there exist positive constants
$\barc < c < \cbar < 1$ and $C$ depending only on the dimension of the manifold,
such that the following properties hold.

First of all, the foliation by subsets $\Hcal_\tau$ of geodesic spheres is defined as follows.
Let $\gamma$ be a future-oriented, timelike geodesic containing $\pbf$ and let us parame\-terize it so that
$\pbf =\gamma(cr)$. Set $\qbf := \gamma(0)$ and consider the (new) observer $(\qbf,\Tbf_\qbf)$ with $\Tbf_\qbf:=\gamma'(0)$.
The constant $\cbar$ is chosen sufficiently small so that the injectivity radius of the map
$\expb_\qbf$ (computed for the observer $(\qbf,\Tbf_\qbf)$) is $\cbar \, r$,  at least.
From the point $\qbf$ we consider normal coordinates $\ybf = (y^\alpha)=(\tau, y^j)$ determined by
the family of future-oriented timelike radial geodesics from $\qbf$, so that
the Lorentzian metric takes the form
$
\gbf = - d\tau^2 + \gbf_{ij} \, dy^idy^j.
$
These coordinates cover a part of the future of the point $\qbf$ and at least the region
$$
\aligned & \Ccal^+(\qbf,\cbar r) := \expb_\qbf\big(C^+(\qbf,\cbar r)\big),
\\
& C^+(\qbf,  \cbar r) := \left\{ V \in B_{\Tbf_\qbf}(0,\cbar r), \quad \gbf_{\Tbf_\qbf}(V,V) < 0,
\quad
{\gbf_{\Tbf_\qbf}( \Tbf_\qbf, V ) \over \, \, \gbf_{\Tbf_\qbf}(V,V)^{1/2}} \geq 1- \cbar \right\}.
\endaligned
$$

The base point $\pbf$ is identified with $(\tau,y^1, \ldots, y^n) = (cr,0,\ldots,0)$ in these coordinates.
By relying on the curvature bound,
analyzing the behavior of Jacobi fields, and using standard 
comparison arguments from Riemannian geometry, one can establish in well-chosen coordinates~\cite{ChenLeFloch}: 
\be
\label{normal}
\aligned
& C^{-1} \delta_{ij} \leq \gbf_{ij} \leq C \, \delta_{ij}, &&
\\
& r^{-1} \Big| \frac{\del \gbf_{ij}}{\del \tau}\Big|
  + r^{-2} \Big| \nabla_{{\del \over \del \tau}} \frac{\del\gbf_{ij}}{\del \tau}\Big| \leq C
&& \quad \text{in } \Ccal^+(\qbf,\cbar r)\cap\big\{ \bartau \leq \tau \leq \taubar \big\},
\endaligned
\ee
where $\bartau := \barc r$ and $\taubar:=\cbar r$.

The reference Riemannian metric associated with
the vector field $\del/\del \tau$ (obtained by parallel transporting the vector $\Tbf_\qbf$) reads
$
\gt : = d\tau^2 + \gbf_{ij} \, dy^idy^j.
$
We use the notation $\dt(\cdot,\cdot)$ and $\St(\cdot,\cdot)$ for the distance
function and the geodesic spheres associated with the metric $\gt$, respectively.
By construction, the function $\tau$ coincides with the distance function $\dt(\qbf,\cdot)$.
It will be useful also to have the following estimate of the Riemann curvature of the Lorentzian metric
\be
\label{curv}
|\Rm|_\gt \leq C \, r^{-2} \qquad \text{in } \Ccal^+(\qbf,\cbar r)\cap\big\{ \bartau \leq \tau \leq \taubar \big\},
\ee
which is based on the reference metric $\gt$.

In turn, the above construction provides us with a foliation (by Lorentzian geodesic spheres)
of some neighborhood
of $\pbf$ (with definite size) by $n$-dimensional spacelike hypersurfaces $\Hcal_\tau$, hence 
$
\pbf \in \bigcup_{\tau \in [\bartau,\taubar]} \Hcal_\tau.
$

Now, consider the time function $\tau$. The standard Hessian comparison
theorem for distance functions in Riemannian geometry is also useful in Lorentzian geometry
and, more precisely, shows that the (restriction of the) Hessian of $\tau$ is equivalent to the induced metric:
\be
\label{Hess}
\bark(\tau, r) \, \gbf_{ij} \leq (-\Dbf^2\tau)|_{E, ij} \leq \kbar(\tau,r) \, \gbf_{ij},
\ee
where $E:= \big( \Dbf \tau \big)^\perp$ denotes the orthogonal complement and
$$
\bark(\tau,r) := \frac{r^{-1} \sqrt{C}}{\tan \big(\tau \, r^{-1} \sqrt{C}\big)},
\qquad
\kbar(\tau,r)
:= \frac{r^{-1} \sqrt{C}}{\tanh \big(\tau \, r^{-1} \sqrt{C}\big)}.
$$
Observe that both $\bark(\tau,r)$ and $\kbar(\tau,r)$ behave like $1/\tau$ when $\tau \to 0$.
Note also that $\bark$ will remain non-singular within the range of interest, since $\tau$ will be chosen to be
a small multiple of $r$.

Consequently, noting that
\be
-\Dbf^2_{ij}\tau = {1 \over 2} \, \frac{\del \gbf_{ij}}{\del \tau} =: A_{ij},
\ee
taking the trace in
the inequalities \eqref{Hess}, and then using the uniform estimate \eqref{normal},
we see that the mean curvature of each slice $\tau=const.$ is close to
$n / \tau$. 
Our objective now is
to replace these slices by constant mean curvature slices by making a small perturbation
determined by solving an elliptic equation in these normal coordinates.

We will also use the {\sl Riemannian geodesic spheres} associated with $\gt$.
Namely, consider the Riemannian distance function $\dt(\pbf', \cdot)$ computed from some arbitrary point
$
\pbf' :=\gamma(\tau) \quad \text{ with } \tau \in [\bartau, \taubar].
$
From the expression of the reference metric we find
\be
\aligned
& \Dbf^2 \dt(\pbf',\cdot) = \nablat^2 \dt(\pbf',\cdot)- 2 \,  \frac{\del \dt}{\del\tau} \, A,
\endaligned
\ee
where $\nablat$ is the covariant derivative associated with $\gt$.
Again by the Hessian comparison theorem and since $\big| \del \dt / \del\tau\big| \leq | \nablat \dt |_\gt =1$,
we find after setting $\widetilde E:={\big(\nablat \dt\big)^\perp}$
\be
\label{hess}
\aligned
\left(\bark(\dt,r) - \frac{C}{r}\right) \, \gt|_{\widetilde E}
\leq \big( \Dbf^2 \dt\big)|_{\widetilde E}
\leq
\left (\kbar(\dt,r) + \frac{C}{r} \right) \, \gt|_{\widetilde E}.
\endaligned
\ee

Choosing now the time variable to be a (small) multiple of $r$ and
taking the trace of the above inequalities, we deduce that for any
$a\in [\barc r, \cbar r]$ the mean curvature $H_{\Acal(\pbf',a)}$
(computed with respect to the ambient Lorentzian metric) of the
(future-oriented, spacelike, and possibly empty) intersection
$\Acal(\pbf',a)$ of the $\gt$-geodesic sphere $\St(\pbf',a)$ and
the future set $\Ccal^+(\qbf,\cbar r)$ satisfies the inequalities
\be 
\label{bard} 
n \, \bark(a,r) \leq H_{\Acal(\pbf',a)}\leq n \, \kbar(a,r), 
\qquad a \in [\barc r, \cbar r]. 
\ee 
Hence, the mean curvature of the Riemannian slices
enjoys the same inequalities as the ones of the Lorentzian slices
$\Hcal_\tau$. Later in this section, we will use the graph of the
Riemannian geodesic spheres as barrier functions.

This completes the discussion of a domain of coordinates $\ybf$
covering a neighborhood of $\pbf$, in which we can assume that all of the above estimates are valid.


\subsubsection*{Mean curvature operator}

We now search for a new foliation $\bigcup_t \Sigma_t$ in which the hypersurfaces
have constant mean curvature
and can be viewed as graphs, say 
$
\Sigma_t := \left \{ G^t(y):= (u^t(y),y) \right\},  
$
over a geodesic leaf $\Hcal_\tau$ for a given value $\tau$ of time-function.
 Here,
$t$ is a real parameter varying in some interval of definite size
and the functions $y \mapsto u^t(y)$ need to be determined. In the following,
we often write $\Sigma = \Sigma_t$, $u=u^t$, and $G=G^t$.
Setting
$u_j := \del u / \del y^j$, the induced metric and its inverse are determined by projection on
the slice $\Sigma$ and read
$$
g_{ij} = \gbf_{ij} - u_i u_j,
\qquad
g^{ij} = \gbf^{ij} + {\gbf^{ik} \gbf^{jl} u_k u_l \over 1 - |\Dbf u|^2},
$$
and the hypersurface $\Sigma$ is Riemannian if and only if
$$
|\Dbf u|^2 = \gbf^{ij}(u,\cdot)u_iu_j < 1.
$$

We are interested here in spacelike hypersurfaces,
 and we denote by $\nabla$ the covariant derivative associated with the induced
Riemannian metric $g_{ij}$. We easily obtain
$$
|\nabla u|^2 = g^{ij} u_i u_j
: =
\frac{|\Dbf u|^2}{1-| \Dbf u|^2},
\qquad |\Dbf u|^2 = {|\nabla u|^2 \over 1 + |\nabla u|^2}.
$$
The future-oriented unit normal to each hypersurface takes the form 
$$
\Nbf = - {\sqrt{1 + |\nabla u|^2}} \, (1, \nabla u).
$$
The second fundamental form of the slice $\Sigma$ is determined by push forward
(with the map $G$) of the coordinate vector fields $Y_j := \del /\del y^j$:
\be
\label{2form}
\aligned
h_{ij}
&: = \big\la \nabla_{G_* Y_i} G_* Y_j, \Nbf \big\ra
\\
& = \frac{1}{\sqrt{1-| \Dbf u|^2}} \, \left( \Dbf_{i}\Dbf_{j}u + {1 \over 2} \, \frac{\del \gbf_{ij}}
 {\del \tau}-{1 \over 2}\gbf^{kl}\frac{\del \gbf_{li}}{\del \tau}u_ku_j
 - {1 \over 2}\gbf^{kl}\frac{\del \gbf_{lj}}{\del \tau}u_ku_i \right)
\\
& = {1 \over \sqrt{1 + |\nabla u|^2}} \, \Big(\nabla_i\nabla_j u+ A_{ij} \Big),
\endaligned
\ee
where
$$
\aligned
& \Dbf_{i}\Dbf_{j} u=\frac{\del^2 u}{\del y^i\del y^j}-\Gammabf^k_{ij}(u,y)\frac{\del u}{\del y^k} 
\endaligned
$$
and $\Gammabf^k_{ij}$ are the Christoffel symbols of $\gbf$.
The tensor field $\Dbf_i \Dbf_j u$ is the spacetime Hessian of the function $u$
(restricted to the hypersurface $\tau=u$), while
$\nabla_i\nabla_j u$ is the fully spatial Hessian defined from the intrinsic metric $g_{ij}$.

The mean curvature of a slice is the trace of $h_{ij}$, that is, in intrinsic form
$$
\Mcal u := h_{ij} g^{ij}
= {1 \over \sqrt{1 +|\nabla u|^2}}
\left(
\Delta u + {A_j}^j\right),
$$
where $\Delta$ is the Laplace operator in the hypersurface, or equivalently in local coordinates
$$ 
\aligned
\Mcal u =
& {1 \over \sqrt{\gbf(u, \cdot)}} {\del \over \del y^i}
\left(\sqrt{\gbf(u, \cdot)} \,  \nu(\Dbf u) \, \gbf^{ij}(u, \cdot) {\del u \over \del y^j} \right)
\\
& + \Big( \nu(\Dbf u)^{-1}  \gbf^{ij}(u, \cdot) + \nu(\Dbf u) \,
\gbf^{ik}(u, \cdot)\gbf^{jl}(u, \cdot) u_k u_l \Big)
{1 \over 2} {\del \gbf_{ij} \over \del \tau}(u, \cdot),
\endaligned
$$
where we have introduced the nonlinear function
$$
\nu(\Dbf u) : = \frac{1}{\sqrt{1-| \Dbf u|^2}} = \sqrt{1 + |\nabla u|^2} = \nu(\nabla u).
$$
Note that, in fact, $\nu$ depends also on $u$.


\subsubsection*{Local formulation of the prescribed mean curvature problem}

We are now ready to introduce a formulation of the problem of interest, in terms of the reference Riemannian metric $\gt$.
Recall that $\gamma$ is a fixed, future-oriented, timelike curve passing through $\pbf$.
Assuming for definiteness that $2c +4 c^2< \cbar$,
from now on we fix some $s \in [c,2c]$ and we introduce the point $\pbf_s := \gamma((s+s^2)r)$
which lies in the future of the base point $\pbf$ since $\gamma$ is a timelike future-oriented curve passing through $\pbf$
for the parameter value $cr$. We then introduce the subset $\Omega_s \subset \{\tau=sr\}$ whose
{\sl boundary is defined} by the condition
$$
\del \Omega_s := \Acal\big( \pbf_s, (s^2+s^3)r \big) \cap \{\tau=sr\}
$$
and which fills up its interior. 
This choice is essential for the mean curvature equation (discussed below)
to admit the Riemannian slices as barrier functions.
Observe that
\be
\label{oumiga}
B_{sr}\big( \gamma(sr), s^{5/2} r/2 \big) \subset \Omega_s \subset B_{sr}\big( \gamma(sr), 2 s^{5/2} r \big).
\ee
Here, $B_{sr}(\gamma(sr), a)$ is the geodesic ball of radius $a$
which lies {\sl in the slice} $\tau=sr$ and is determined by the metric $\gbf_{ij}$ induced on the geodesic leaf.

Finally, given $\alpha \in (0,1)$ and a bounded function $H$ of class $C^\alpha$ defined on $\Omega_s$ and satisfying
the  restriction
$$
n\kbar(sr,r) \leq H \leq n\bark(2s^2r,r),
$$
we seek for a spacelike hypersurface with mean curvature $H$ and boundary
$\del \Omega_s$. Analytically, this is equivalent to solving the Dirichlet problem
\be
\label{dirichlet}
\aligned
\Mcal u & = H \qquad \text{ in } \Omega_s,
\\
u & = sr \qquad  \text{ in } \del \Omega_s,
\endaligned
\ee
in which, therefore, we have prescribed both the boundary of the unknown hypersurface and its mean curvature.
In the present paper, we are mainly interested in the case that $H$ is a constant function.
We also assume that $s$ is sufficiently small.


\subsection{Statements of the uniform estimates}

To establish the existence of CMC hypersurfaces as graphs over a given geodesic leaf $\Hcal_\tau$,
the main difficulty is to bound
$|\Dbf u|$ away from $1$,
for all functions $u$ satisfying and $u=\tau$ on the boundary and
having their mean curvature pinched in some interval.
Precisely, the rest of this section is devoted to the proof of the following result.

Recall that $(\Mbf, \gbf,\pbf,\Tbf_\pbf)$ denotes a pointed Lorentzian
manifold satisfying the curvature and injectivity radius assumptions \eqref{bd3} at some scale $r>0$ and that
$\gamma$ is a future-oriented timelike geodesic satisfying $\gamma(cr)=\pbf$.

\begin{proposition}[Uniform estimates for CMC hypersurfaces]
\label{folia}
There exist constants $c,\theta>0$ depending on the dimension $n$ only such that the following property holds
with the notation introduced in this section.
For any $s \in [c,2c]$ and $t\in [ n\kbar(sr,r), n\bark(2s^2r,r)]$ there
exists a solution $u$ to the Dirichlet problem \eqref{dirichlet} associated
with the (constant) mean curvature function $H \equiv t$ and such that
$$
\sup_{\Omega_s} |\Dbf u| \leq 1 - \theta, \qquad \sup_{\Omega_s'} r |h| \leq \theta^{-1},
$$
where $\Omega_s'=B_{sr}\big( \gamma(sr),s^{5/2} r/4 \big) \subset \{\tau=sr\}$. 
\end{proposition}

Observe that the bound on the second fundamental form holds only in a subset of $\Omega_s$,
whose diameter, however, is also of the order $r$.
Theorem~\ref{foli} is immediate once we establish Proposition~\ref{folia}.
The proof of Proposition~\ref{folia} will follow from several preliminary results.
The first lemma below is a direct consequence of the maximum principle for elliptic operators. The other
lemmas will be established in Subsection~\ref{proofs}.

\begin{lemma}[Comparison principle]
\label{compa}
Given two functions $u,w$ satisfying $\Mcal u\geq \Mcal w$ in their domain of definition
and $u\leq w$ along the boundary, one has either
$u<w$ in the interior of their domain of definition or else $u\equiv w$. In particular,
if $\Mcal u \geq n \, \kbar(\taubar, r)$ everywhere
and $u\leq \taubar$ along the boundary, then $u \leq \taubar$.
Similarly, if $\Mcal u\leq n\, \bark(\bartau, r)$ everywhere and $u\geq \bartau$ along the boundary, then $u\geq \bartau$.
\end{lemma}

\begin{lemma}[Boundary gradient estimate]
\label{boundary}
For any solution $u$ of \eqref{dirichlet} with mean curvature function satisfying
$n \, \kbar(sr,r) \leq \Mcal u \leq n \, \bark(2s^2r,r)$ one has
\be
\label{48}
| \Dbf u| 
< {1 \over 2} \quad \text{ on the boundary } \del \Omega_s.
\ee
\end{lemma}

\begin{lemma}[Global gradient estimate]
\label{global}
Under the assumptions of Lemma~\ref{boundary} one has
$$
\sup_{\Omega_s} |\nu( \Dbf u)| \leq C_1(n),
$$
where the constant $C_1(n)$ depends on the dimension, only.
\end{lemma}

Now, in view of Lemmas~\ref{boundary} and \ref{global} and by standard arguments \cite{GilbargTrudinger},
one can check that for each $t\in [\kbar(sr,r),\bark(2s^2r,r)]$ the
Dirichlet problem \eqref{dirichlet} admits a smooth solution $u$
determining a slice with constant mean curvature $t$.
 
Note that, by Lemma~\ref{global}, the induced metric on
$\Sigma_t$ is equivalent to the metric $\gbf_{ij}$ on the domain $\Omega$ so that
we can use, for instance, Sobolev inequalities on $\Sigma_t$.

\begin{lemma}[Interior estimates for the second fundamental form]
\label{sec}
Under the assumptions of Lemma~\ref{boundary}, for all $q\in [1,\infty)$ there exist
positive constants $C_2(n)$ and $C_3(n,q)$ such that for every $p' \in \Sigma \setminus \del \Sigma$ 
$$ 
\aligned
& |h(p') | \leq \frac{C_2(n)}{d(p', \del \Sigma)},
\\
& \left(\frac{1}{d(p', \del \Sigma)^n}\int_{B\big(p',d(p', \del \Sigma)/4\big)} |\nabla h|^q \, dv_\Sigma\right)^{1/q}
\leq \frac{C_3(n,q)}{d(p', \del \Sigma)^2},
\endaligned
$$
where $dv_\Sigma$ is the induced volume form on $\Sigma$ and
$d(p',\del \Sigma)$ is the distance to the boundary $\del \Sigma$
associated with the induced metric $g_{ij}$ on $\Sigma$.
\end{lemma}

Observe that the upper bound in the above lemma blows-up if the point $p'$ approaches the boundary of the CMC slice,
and that for $p' \in \Omega_s'$ the factor $d(p', \del \Sigma)$ is of order $r$, as required for
Proposition~\ref{folia}.

\begin{lemma}[Time-derivative of the level function]
\label{tidu}
Under the assumptions of Lemma~\ref{boundary} there
exist constants $C_4(n), C_5(n)>0$ such that
$$ 
C_4(n) r^{2} \leq - {\del u \over \del t} \leq C_5(n) r^{2}
 \qquad \text{ on } \Omega_s'.
$$
\end{lemma}


\subsection{Derivation of the uniform estimates}
\label{proofs}

\begin{proof}[Proof of Lemma~\ref{boundary}] We use here the maximum principle stated in Lemma~\ref{compa}.
The part of the Riemannian geodesic sphere
$\St\big(\pbf_s,(s^2+s^3)r\big)$ (defined by $\gt$) ``below'' $\Omega_s$,
that is the part corresponding to $\tau\leq sr$,
is the graph $y \mapsto (\underline{u}(y),y)$ of a function $\underline{u}$ over $\Omega_s$ whose boundary values
are $sr$ on $\del\Omega_s$. Since $s$ is sufficiently small, one easily checks that, for instance,
\be
\label{grabaru}
| \Dbf \baru|< s^{1/2} \qquad \text{ on } \Omega_s
\ee
and, in particular, $\baru$ satisfies \eqref{48} along the boundary $\del \Omega_s$.

Suppose now that there exists a $C^2$ spacelike hypersurface $(u(y),y)$
defined over $\Omega_s$ having the same boundary values as the function $\underline{u}$
and such that its mean curvature $H$ remains bounded in the interval
$[n\kbar(sr,r), n\bark((s^2+s^{5/2})r,r)]$. Let us set
$$
\mbar = \sup_{y\in \Omega_s}u(y), \qquad
\barm = \sup_{y\in \Omega_s} \dt((u(y),y),\pbf_s)
$$
and use the following comparison technique.

Note that the range of the function $u$ lies between $(s-2s^3)r$ and $(s+2s^3)r$,
since $u$ is spacelike (so the norm of its gradient can not exceed $1$),
its boundary value is $sr$, and the diameter of the set $\Omega_s$ is $4s^3r$ at most.
We are going to show the pinching property
$$
\baru-sr\leq u-sr\leq 0,
$$
which immediately implies the desired boundary gradient estimate \eqref{48}.

First of all, we claim that $\mbar=sr$. If this were not true, then the maximum of $u$ would be achieved at some
point $y_0$ in the interior of $\Omega_s$. Since the graph of $u$ is below the graph of $\tau\equiv \mbar$
and both graphs are tangent at the point $(y_0,\mbar)$, we conclude that at the point
$(y_0,u(y_0))$ the mean curvature of $u$ is less or equal to that of $\tau\equiv \mbar$.
However, in view of the Hessian estimate \eqref{Hess} this is a contradiction if $\mbar>sr$.

Considering next the lower bound for $u$, we claim that $u \geq \underline {u}$ on $\Omega_s$.
Otherwise, by contradiction there would exist a point $y_1 \in \Omega_s$ such that (at least)
$$
\dt((u(y_1),y_1),\pbf_s) = \barm < (s^2+2s^3)r.
$$
By comparing, at the base point $(u(y_1),y_1)$,
the mean curvature of the graph $u$ and the one of the sphere $S\big(\qbf',\dt((u(y_1),y_1),\pbf_s) \big)$,
 we find that
$$
\Mcal u(y_1)> n\bark((s^2+2s^3)r,r),
$$
which contradicts our assumption $\Mcal u \leq n\bark(2s^2r,r)$.
\end{proof}

\begin{proof}[Proof of Lemma~\ref{global}] {\it Step 1.} We will first show that, for some
 sufficiently
large $p$, the sup norm of $\nu(\Dbf u)$ is
bounded by its $L^p$ norm. By scaling, we may assume $r=1$ from
now on. Here, as in the rest of this paper, the main difficulty is
making sure that all constants arising in the following arguments
depend on the injectivity radius and curvature bounds, only.   It
will be convenient to work with the intrinsic form of the
mean-curvature operator $\Mcal$, but the expression in coordinates
will be also used in the end of the argument in order to control
certain Sobolev constants.

Recall that, on the hypersurface $\Sigma$,
\be
\label{equa3}
\Delta u + A_j^j = \nu(\nabla u) \, H,
\ee
where $\Delta$ denotes the Laplace operator
on the slice and $H$ is the prescribed mean curvature
function. Observe that since the second fundamental form of the
geodesic sphere is bounded, we have
$$
|A_j^j| = \left| g^{ij}\frac{\partial \textbf{g}_{ij}}{\partial \tau} \right|
\leq
C \, g^{ij} \gbf_{ij}\leq \frac{C}{1-|\Dbf u|^2},
$$
hence
\be
\label{est3}
|\Delta  u | \leq C \, |\nu(\nabla u)|^2.
\ee
Observe that the coefficients of the Laplace operator on $\Sigma$
are nothing but metric coefficients on which, at this stage of the analysis,
we have an $L^\infty$ control, only.

We are going to use the classical Weitzenb\"ock identity applied to the level
function~$u$
\be
\label{boch} \Delta |\nabla u|^2
= 2 \, |\nabla^2 u|^2 + 2 \, \la \nabla u, \nabla \Delta u\ra + 2 \, \Ric(\nabla u, \nabla u).
\ee
First, recalling that $G= (u,y)$ we estimate the Ricci curvature term by relying on the Gauss formula
\be
\nonumber
\aligned
R_{ijkl} = &
\Rbf_{\alpha\beta\gamma\delta}G^{\alpha}_iG^{\beta}_jG^{\gamma}_kG^{\delta}_{l}
- \big( h_{ik} h_{jl} - h_{il} h_{jk} \big)\\
=
& \textbf{R}_{ijkl}+\textbf{R}_{0jkl}u_i+\textbf{R}_{i0kl}u_j+\textbf{R}_{ij0l}u_k+\textbf{R}_{ijk0}u_l
    + u_iu_k\textbf{R}_{0j0l}
\\
& +u_iu_l\textbf{R}_{0jk0} +u_ju_k\textbf{R}_{i00l}+u_ju_l\textbf{R}_{i0k0}- \big( h_{ik}
h_{jl} - h_{il} h_{jk} \big),
\endaligned
\ee
where $G_{i}=G_{*}(\frac{\partial}{\partial
y^i})=\frac{\partial}{\partial y^i}+\frac{\partial u}{\partial
y^{i}}\frac{\partial}{\partial \tau}.$
 The spacetime curvature being uniformly bounded, we find
$$
\aligned
\Rbf_{\alpha\beta\gamma\delta} G^\alpha_i G^\beta_j G^\gamma_k G^\delta_l g^{jl}
& \geq -C \, (1+ |\nabla u|^2) \, \gbf_{ik}
\\
& = - C \, (1+ |\nabla u|^2) \, (g_{ik} + u_i u_k).
\endaligned
$$
Taking the trace of the Gauss formula, we obtain a lower bound
for the Ricci curvature of the hypersurface:
$$
\aligned R_{ik} &\geq h_{il}h_{kj}g^{lj}-H h_{ik}-C (1+ |\nabla u|^2)(g_{ik}+u_iu_k),
\endaligned
$$
and therefore
$$
\aligned \Ric(\nabla u, \nabla u) &\geq -C  (1+ |\nabla u|^2)^3.
\endaligned
$$
In turn, from \eqref{boch}  we deduce the key inequality
\be
\label{key3}
\aligned
& \Delta |\nabla u|^2 - 2 |\nabla^2 u|^2
 \geq 2 \, \la \nabla u, \nabla \left( \Delta u\right) \ra
      - C \, \big(1+ |\nabla u|^2\big)^3,
\endaligned
\ee
which is an {\sl intrinsic} statement written on the hypersurface $\Sigma$ and the constant $C$
depends on the spacetime curvature bound, only.

Next, to estimate the gradient of $u$ we consider the function
$$
v=v(\nabla u) := (1+|\nabla u|^2 - k)_+,
$$
where $k$ is chosen suitably large so that, thanks to the boundary
gradient estimate in Lemma~\ref{boundary}, the function $v(|\nabla
u|)$ vanishes on the boundary of the hypersurface $\Sigma$.
Multiplying \eqref{key3} by $v^q$ for $q \geq 1$ integrating
over the hypersurface, and using Green's formula we obtain
$$
\aligned
& \int_\Sigma \Big( q \, v^{q-1} |\nabla v|^2 + 2 v^q
\, |\nabla^2 u|^2 \Big) \, dv_\Sigma
\\
& \leq \int_\Sigma \Big( 2 q \, v^{q-1} \la \nabla v, \nabla u\ra \, \Delta u
                            + 2 \, v^q \, |\Delta u|^2 + C \, (v^{q+3}+v^q) \Big) \,  dv_\Sigma.
\endaligned
$$

At this juncture, we observe that the higher-order term $\Delta u$
is controlled by the prescribed mean curvature equation
\eqref{est3}. We obtain
$$
\aligned & \int_\Sigma \Big( q \, v^{q-1} |\nabla v|^2 + 2 v^q \,
|\nabla^2 u|^2 \Big) \, dv_\Sigma
 \leq (q+1) \, C \int_\Sigma \left( v^{q+3}+v^{q-1} \right) \,
dv_\Sigma
\endaligned
$$
and thus
\be
\label{key5}
\int_\Sigma \left| \nabla (v^{(q+1)/2}) \right|^2 \, dv_\Sigma
\leq
(q+1)^2 \, C \int_\Sigma \big( v^{q+3}+v^{q-1}\big) \, dv_\Sigma.
 \ee

To make use of \eqref{key5} it is convenient to return to our notation in coordinates,
by observing that
$$
\aligned
\sqrt{det(g)} & = \sqrt{1-|\Dbf u|^2} \, \sqrt{det(\gbf)},
\\
|\nabla (v^{(q+1)/2})|^2 & \geq \big(\frac{q+1}{2}\big)^2 v^{q-1}
|\Dbf v|^2,
\endaligned
$$
so that
$$
C \, |\nabla (v^{(q+1)/2})|^2 \, \sqrt{det(g)} \geq |\Dbf
v^{(q+\frac{1}{2})/2}|^2 \, \sqrt{det(\gbf)},
$$
where we used (if $v>0$)
$$
C \, v^{-1/2} \geq \sqrt{1-|\Dbf u|^2} \geq C' \, v^{-1/2}.
$$

Therefore, provided we now assume that $q \geq 3/2$, \eqref{key5} takes
the following coordinate-dependent form:
\be
\label{key5b}
\int_{\Omega_s} |\Dbf
(v^{(q+\frac{1}{2})/2})|^2 \, dy  \leq (q+1)^2 C \int_{\Omega_s} \big( v^{q+3 - 1/2} + v^{q-1-1/2} \big)\, dy.
\ee
By Sobolev's inequality in the local coordinates under consideration, we have
$$
\left( \int_{\Omega_s} w^{2n/(n-1)} \, dy \right)^{(n-1)/n}
\leq
C \int_{\Omega_s} \big( |\Dbf  w|^2 + w^2 \big) \, dy,
$$
which we apply to the function $w := v^{(q+1/2)/2}$.
Recalling that $r=1$ (after normalization) and observing that the domain of integration in $y$
is bounded,
we deduce from \eqref{key5b} that
$$
\left( \int_{\Omega_s} v^{(q+\frac{1}{2})n/(n-1)} \, dy \right)^{(n-1)/n}
\leq C \, (q+1)^2 \int_{\Omega_s} \big( v^{q+3-1/2} + v^{q-1-1/2}\big) \, dy
$$
or, equivalently, for all $p>2$
\be
\label{key5c}
\left( \int_{\Omega_s} v^{pn/(n-1)} \, dy \right)^{(n-1)/(pn)}
\leq
C^{1/p} \, p^{2/p} \left(\int_{\Omega_s} \big( v^{p+2}+v^{p-2} \big) \, dy\right)^{1/p}.
\ee

Without loss of generality, we may assume that $\| v\|_{L^\infty(\Omega_s)}  \geq 1$, for
otherwise the result is immediate. Then, \eqref{key5c} leads to the main estimate
$$
\aligned
& \max\Big(1, \left( \int_{\Omega_s} v^{pn/(n-1)} \, dy \right)^{(n-1)/(pn)} \Big)
\\
&
\leq C^{1/p} \, p^{2/p} \| v\|_{L^\infty(\Omega_s)}^{2/p}  \max\Big( 1, \left( \int_{\Omega_s} v^p \, dy\right)^{1/p}\Big).
\endaligned
$$
It remains to iterate the above estimate, which yields
$$
\aligned
& \| v\|_{L^\infty(\Omega_s)}
\leq C' \, \| v\|_{L^\infty(\Omega_s)}^\alpha  \left( \int_{\Omega_s} v^{p_0} \, dy \right)^{1/{p_0}},
\\
& \alpha:= {2 \over p_0} \sum_{k=0}^\infty (1-1/n)^k = {2n\over p_0}.
\endaligned
$$
In conclusion, provided that $p_0>2n$ the sup norm of $v$ is uniformly
bounded by its $L^{p_0}$ norm.

\

\noindent{\it Step 2.} It remains to derive an estimate for some $L^{p_0}$ norm.
Following \cite{Gerhardt}, we return to the
inequality \eqref{est3} satisfied by the function $u$ and, for every $\lambda$,
we write
$$
\aligned
\Delta (e^{\lambda \, u})
& = \lambda e^{\lambda \, u} \Delta u + \lambda^2 e^{\lambda \, u} |\nabla u|^2
\\
& \geq
- C \, \lambda \, e^{\lambda \, u}  (\nu(\nabla u))^2 + \lambda^2 e^{\lambda \, u} |\nabla u|^2.
\endaligned
$$ 
Combining this estimate with a direct calculation from \eqref{key3} (similar to the one in Step~1 above),
we obtain
$$
\aligned
\Delta \left( v^q e^{\lambda \, u} \right)
\geq
& \lambda^2 \, v^{q+1} \, e^{\lambda \, u} + \lambda v^{q-1} e^{\lambda u} \Big(
\la \nabla u, \nabla v \ra - C v \, (v+1) \Big)
\\
& + q v^{q-1} e^{\lambda u} \Big( 2 |\nabla^2 u|^2 - C v^3 + 2 \la \nabla u, \nabla (\Delta u) \ra \Big)
\\
& + q (q-1) v^{q-2} e^{\lambda u} |\nabla v|^2.
\endaligned
$$
Then, by integrating over the hypersurface $\Sigma$, integrating by parts, using
\eqref{est3}  to control the term $\Delta u$, and finally
choosing $\lambda$ sufficiently large, we arrive at
$$
\int_\Sigma |\nabla u|^q \, dv_\Sigma \leq C_q'.
$$
This completes the proof of Lemma~\ref{global}.
\end{proof}

\begin{proof}[Proof of Lemma~\ref{sec}]
{\it Step 1.} We are going to control the sup norm of $h$, and to this end we
will use Nash-Moser's iteration technique.
Note that the elliptic equation satisfied by the second fundamental form a~priori
has solely $L^\infty$ coefficients.
By scaling we can assume $r=1$. We consider an arbitrary point $p' \in \Sigma$ and
we set $\delta:=d(p',\del \Sigma)$. Simons' identity \cite{Simons} for the hypersurface reads
\be
\label{simon1}
\aligned
&\Delta h_{ij} = \Delta h_{ij} - (tr h)_{ij}
\\
& = |h|^2 \, h_{ij}-(tr h)
h_{ik}h_{lj}g^{kl}-\Rbf_{ipjq}h_{kl}g^{pk}g^{ql}+\Rbf_{jplq}h_{ik}g^{pq}g^{kl}
\\
& \quad +\nabla_{p}(\Rbf_{qj\Nbf i})g^{pq} -
\nabla_{j}(\Rbf_{i\Nbf}),
\endaligned
\ee
in which the Hessian $(tr h)_{ij}$ vanishes since $\Sigma$ has constant mean curvature.
Recall here that $\Nbf$ is the future-oriented normal to $\Sigma$.
Thanks to \eqref{simon1} we obtain
\be
\label{simon2}
\aligned
\Delta |h|^2
\geq \,
& 2 \, |\nabla h|^2 + 2 \, |h|^4 - C \, (|\Rm|_\Nbf +1) \, |h|^3
\\
& + 2 \, \la\nabla_p (\Rbf_{qj\Nbf i})g^{pq} - \nabla_j (\Rbf_{i\Nbf}), h_{ij}\ra.
\endaligned
\ee

Let $\varphi$ be a smooth, non-negative, non-increasing cut-off
function which equals $1$ in the interval $[0,1/2]$ and $0$ in
$[1,\infty)$.  Then, the function
$\psi := \varphi \circ \kappa$ with $\kappa := (d(p',\cdot) /\delta)$ is a
cut-off function on the CMC hypersurface $\Sigma$ which vanishes near the boundary $\del \Sigma$.

Fix some $q \in [1,\infty)$. Multiplying (\ref{simon2}) by $\psi|h|^q$, integrating over $\Sigma$,
and then integrating by parts, we arrive at
$$
\aligned
& \int_\Sigma \psi |h|^q \Delta |h|^2  \, dv_\Sigma
\\
& \geq \int_\Sigma \Big(
\psi |h|^q\Big( 2|\nabla h|^2 + 2 \, |h|^4 - C \, |h|^3 -C \, |\textbf{Rm}|_\Nbf (q+1) \, |\nabla h|\Big)
\\
&
 \qquad \qquad - C \, |\textbf{Rm}|_\Nbf \, |\nabla \psi| \, |h|^{q+1} \Big) \, dv_\Sigma.
\endaligned
$$
Using
$$
\int_\Sigma \psi \, |h|^q \, \Delta |h|^2  \, dv_\Sigma
\leq
\int_\Sigma 2 \, |\nabla \psi| \, |h|^{q+1} \, |\nabla h| \, dv_\Sigma
  - \int_\Sigma 2 q \, \psi \, |h|^q \, |\nabla |h||^2 \, dv_\Sigma
$$
and Cauchy-Schwartz's inequality, we obtain
\be
\label{int2'}
\aligned
& \int_\Sigma \psi \, |\nabla h|^2 \, |h|^q + \int_\Sigma \psi \, |h|^{q+4}  \, dv_\Sigma
\\
&
\leq C \, \int_\Sigma \left( \left( \frac{|\varphi'|^2}{\delta^2 \varphi} + \varphi \right)
\circ \kappa \, |h|^{q+2}
 + (q+1)^2 \, \psi \, |h|^q
 + \frac{1}{\delta} \, \left| \varphi' \circ \kappa \right|  \, |h|^{q+1} \right)  \, dv_\Sigma,
\endaligned
\ee
in which we can always choose $\varphi$ so that $|\varphi'|^2\leq C \, |\varphi|$.

Then, by H\"older's inequality we have
\be
\label{iter}
\aligned
\left( \int_\Sigma \psi \, |h|^{q+4}  \, dv_\Sigma \right)^{1/(q+4)}
& \leq C_q \, \delta^{\frac{n}{q+4}-1}.
\endaligned
\ee
In view of Lemma \ref{global}, the hypersurface is uniformly spacelike
and so we have the Sobolev inequality
$$
\aligned
\left( \int_\Sigma|\psi \, |h|^{q+2}|^{\frac{n}{n-1}} dv_\Sigma \right)^{\frac{n-1}{n}} 
\leq
C' \, \int_\Sigma \big| \nabla(\psi |h|^{q+2}) \big| \, dv_\Sigma.
\endaligned
$$
Combining  this with (\ref{int2'}) and suitably choosing the function $\varphi$, we find
that for all $i=1,2,\ldots$
\be
 \aligned
& \left(
\int_{B(p',\frac{\delta}{2}+\frac{\delta }{2^{i+1}})} \, |h|^{\frac{n(q+2)}{n-1}} \, dv_\Sigma \right)^{\frac{n-1}{n(q+2)}}
\\
& \leq
\left(2^i\frac{C' (2+q)^2}{\delta}\right)^{\frac{1}{q+2}}
\left( \int_{B(p',\frac{\delta}{2}+\frac{\delta }{2^i})} |h|^{q+2}\, dv_\Sigma \right)^{\frac{1}{q+2}}.
\endaligned
\ee
Using the Nash-Moser's iteration technique we deduce that
\be
\label{Moser}
\sup_{B(p',\delta/2)} |h|
\leq C_q \, \delta^{-\frac{n}{q+2}} \left( \int_{B(p',3\delta/4)} |h|^{q+2}
      \, dv_\Sigma\right)^{\frac{1}{q+2}}.
\ee
Finally, choosing $q=0$ in (\ref{iter}) and $q=2$ in \eqref{Moser}, we find
$$
\sup_{B(p',\delta/2)} |h| \leq \frac{C''}{\delta}.
$$

\noindent{\it Step 2.} Next, by relying on the sup norm estimate that we just established, we can
estimate the covariant derivative of $h$. We need an $L^p$ estimate for the equation \eqref{simon1}.
From Gauss equation we see that the curvature of the hypersurface is bounded by $C' \delta^{-2}$ on
the ball $B(p',\delta/2)$. By introducing harmonic coordinates on the (Riemannian) slice $\Sigma$,
as in \cite{JostKarcher}, we see that the metric coefficients belong to the H\"older space $C^{1,\alpha}$.
Since the right-hand side of \eqref{simon1}
belongs to the Sobolev space $W^{-1,q}$ for any $q \in (1, \infty)$, thanks to the Sobolev regularity property
for elliptic operators (in fixed coordinates) we find
$$
\left( \frac{1}{\delta^n}\int_{B(p',\delta/4)}|\nabla h|^q dv_\Sigma \right)^{1/q}
\leq \frac{C_q}{\delta^2}
$$
for some constant $C_q>0$,
which completes the proof of Lemma~\ref{sec}.
\end{proof}


\begin{proof}[Proof of Lemma~\ref{tidu}] We need now to estimate the time-derivative of the level set function $u$.
Set $\pbf'' = (u(\gamma(sr)),\gamma(sr)) \in \Sigma$ and consider the geodesic distance function $\rho=\rho(\pbf'',\cdot)$
associated with the {\sl induced metric} on the CMC hypersurface $\Sigma$. By the Gauss equation,
the Ricci curvature is bounded (especially from below) by
$$
R_{ij} \geq -\frac{C'}{r^2} \, g_{ij}.
$$
Hence, thanks to the Laplacian comparison theorem, the distance function $\rho$ is a supersolution for
the operator $-\Delta + C'/\rho$, that is, in the weak sense
$$
\Delta \rho\leq \frac{C'}{\rho}.
$$

Let $\varphi$ be the (non-increasing) cut-off function introduced in the proof of
Lemma~\ref{sec}. Then,
by differentiating with respect to $t$ the equation \eqref{dirichlet} satisfied by the solution $u$
and in view of the bounds on $h$ in Lemma~\ref{sec},
we obtain
\be
\label{ubb}
\Big( \Delta - |h|^2 - \Ricbf(\Nbf,\Nbf) \Big)
\Big( \nu(\nabla u) \, \frac{\del u}{\del t} + \eps \, \varphi\Big(\frac{4\rho}{s^{5/2} r}\Big) \Big)
\geq
1-\eps \frac{C'}{r}\geq 0,
\ee
where we have set $\eps := r^2 / C'$. Finally, applying the maximum principle (see \eqref{NNN} below for the ellipticity property)
we conclude that
$$
- C' \, r^2 \leq \frac{\del u}{\del t} \leq - \frac{1}{C'} \, r^2
\qquad
\text{ on } \Omega_s '= B_{sr}\Big(\gamma(sr),s^{5/2} r/4\Big)\subset\{\tau=sr\}.
$$
\end{proof}


\begin{proof}[Proof of Proposition~\ref{folia} and Theorem \ref{foli}]
{\it Step 1.}
The first variation $\Lcal\Mcal (X)$ of the mean curvature along an arbitrary vector field $X$ reads
\be
\label{implicit}
\Lcal\Mcal (X) = \Delta \la X,\Nbf\ra - \big(|h|^2 + \Ricbf(\Nbf,\Nbf) \big) \, \la X,\Nbf\ra
- \la X,\nabla H\ra,
\ee
where we recall that $\Nbf$ is the unit normal vector field to the hypersurface and $h$ is its second fundamental form.
In the case of graphs, the linearization of the mean curvature equation around
a {\sl constant} mean curvature hypersurface reads
$$
\Lcal\Mcal (\varphi)=\Delta \big( \nu(\nabla u) \, \varphi \big)
      - \big( |h|^2 + \Ricbf(\Nbf,\Nbf) \big) \, \nu(\nabla u) \, \varphi.
$$
Under our assumptions, this operator is uniformly invertible, since
\be
\label{NNN}
|\Ricbf(\Nbf,\Nbf)|\leq C(n) r^{-2}, \qquad
|h|^2\geq \frac{H^2}{n}\geq s^{-1/2} r^{-2}
\ee
and we choose $s$ sufficiently small so that the term $|h|^2$ dominates $\Ricbf(\Nbf,\Nbf)$.

More precisely, we have the following important conclusion.

By the implicit function theorem, the linearized mean curvature operator $\Lcal\Mcal$
on the space of all
spacelike $C^{2,\alpha}$ functions $u$, with fixed boundary value $cr$ on $\del \Omega_s$
and with $\alpha \in (0,1)$,
is locally invertible around any smooth hypersurface of constant mean curvature.
In consequence, starting from any fixed spacelike hypersurface $u$ with constant mean curvature
$$
2\kbar(sr,r)\in [\kbar(sr,r),\bark(2s^2r,r)]
$$
and using the implicit
function theorem, we can find a smooth family of spacelike hypersurfaces $u^t$ with constant mean curvature $t$
varying in an $\eps$-neighborhood of $2\kbar(sr,r)$. From Lemma~\ref{global} and Schauder's estimate, we then
deduce higher-order uniform estimates for $u^t$. So, the function $u^t$ is smooth and we can take
a convergent subsequence of values $t$ converging to the end-points of the interval.

Next, by continuation we may still use the implicit function theorem and extend
 the smooth family under consideration
for all mean curvature parameter values in the interval $[\kbar(sr,r),\bark(2s^2r,r)]$. In other words,
we conclude that for any $t\in [\kbar(sr,r),\bark(2s^2r,r)]$ we can solve the Dirichlet
problem \eqref{dirichlet} of prescribed mean curvature equal to $t$,
and, moreover, the solution $u$ depends smoothly upon the mean curvature $t$.

Next, differentiating with respect to $t$ the conditions satisfied by the solution $u$  we obtain
\be
\label{ub}
\aligned
\Big( \Delta - \big(|h|^2 + \Ricbf(\Nbf,\Nbf)\big) \Big)\Big( \nu(\nabla u) \, \frac{\del u}{\del t} \Big)
     & = 1 \qquad \text{ in } \Omega_s,
\\
\frac{\del u}{\del t} & = 0 \qquad \text{ along } \del\Omega_s.
\endaligned
\ee
As already observed, by the maximum principle we have
$$
- C r^2 < \frac{\del u}{\del t} < 0.
$$
This property shows that the family of CMC hypersurfaces forms a foliation of the
region under consideration.

\

\noindent{\it Step 2.} In this last part of the construction we choose the geodesic slice over which the CMC foliation
should be based.
We observe that the foliation constructed in Step~1
need not pass through the given observer $\pbf=\gamma(cr)$.
To cope with this difficulty, we now vary the parameter $s$
in order to ensure that the foliation contains a neighborhood of $\pbf$.
We proceed as follows.

For any $s\in [c,2c]$ we use the notation $u^{(s)}$ for the function describing the CMC hypersurface
constructed over the reference domain $\Omega_s\subset \{\tau=sr\}$ for the chosen value of the mean curvature
$$
2\kbar(sr,r)\in  [\kbar(sr,r),\bark(2s^2r,r)].
$$
Now, we emphasize that
$u^{(s)}$ depends continuously upon the parameter $s$. Indeed, in view of \eqref{implicit} we can apply the
implicit function theorem and we see that the solution depends smoothly upon the parameters arising in the domain of definition and upon the boundary values.
Therefore, recalling the result in Lemma~\ref{global}, given any function $u^{(c)}$ we may extend it to a whole
family $u^{(s)}$ smoothly for all $s \in [c,2c]$.

Then, by Lemma~\ref{compa} we have $ u^{(s)} \geq \tau(s)r$ for some $\tau(s)$ satisfying
$
2\kbar(sr,r)=\bark(\tau(s)r,r),
$
which implies that $\tau(2c)>c$, at least when $c$ is suitably small.
Since we have
$u^{(c)}(\gamma(cr)) < cr$, by continuity there is some $s_0 \in [c,2c]$ such that
$
u^{(s_0)}(\gamma(cr)) = cr.
$
Therefore, we have constructed a family of CMC hypersurfaces $u^t$ with
constant mean curvature 
$$
t \in [n \kbar(s_0r,r), n\bark(2s_0^2r,r)] \qquad \text{ for some } s_0\in[c,2c]
$$
over some geodesic slice $\Omega_{s_0}\subset \{\tau=s_0r\}$.
Most importantly, the point $\gamma(cr)$ lies in the CMC hypersurface
with mean curvature $2\kbar(s_0r,r)$.

In addition, by a direct computation we obtain
$$
\Dbf t = \sqrt{1-|\Dbf u|^2} \, \Big(\frac{\del u}{\del t}\Big)^{-1} \frac{(-1,\nabla u)}{\sqrt{1-|\Dbf u|^2}}
$$
and, in view of Lemma~\ref{tidu}, the proof is now completed.
\end{proof}


\subsection{Further geometric estimates}

For any $p' \in \Sigma_t$ with $\delta r = d(p',\del \Sigma_t)$
and thanks to our estimate of the second fundamental form in $B(p',\delta/2)$ and Gauss equation,
we see that the curvature of the hypersurface is bounded by $C \, \delta^{-2} r^{-2}$. By choosing harmonic
coordinates as in \cite{JostKarcher} and using the $L^p$ estimates in
(\ref{ub}), the function $\lambda= - \nu(\nabla u) \, \frac{\del u}{\del t}$ satisfies (for any $q \in [1, \infty)$)
\be
\label{lambdatwo}
\left( \frac{1}{\delta^n r^n} \int_{B(\pbf',\delta/4)} |\nabla ^2 \lambda|^q \, dv_\Sigma\right)^{1/q}
\leq \frac{C_q}{\delta ^2 r^2}.
\ee

In addition, let us investigate the geometry of the boundary $\del \Sigma$ of the foliation leaves. More precisely,
we can estimate its second fundamental form $\II_{\del \Sigma}$, as follows.

\begin{proposition}[Boundary of the CMC foliation]
\label{bcmc}
The CMC hypersurfaces constructed in the proof of Theorem~\ref{foli}
also satisfy the uniform estimate
$$ 
|\II_{\del \Sigma}| \leq \frac{C(n)}{r}.
$$
\end{proposition}

\begin{proof} Recall that, for any tangent vector fields $X,Y$ along $\del \Sigma$, the scalar
$\II_{\del \Sigma}(X,Y)$ is defined as $\gbf(\Dbf_{X}Y,N_{\del \Sigma})$, where  $N_{\del \Sigma}$ is the
normal vector field of $\del \Sigma$ in the hypersurface $\Sigma$.

On the other hand, since $\del \Sigma=\del \Omega$ is obtained by the intersection of two level surfaces
$\Hcal:=\{\tau = const.\} $ and $\St:=\{\dt=const.\}$, $\del \Sigma$ may be regarded
as a hypersurface of codimension $1$ in either $\Hcal$ or $\St$.
The second fundamental form $\II_{\del \Sigma}^{\Hcal}$ of $\del \Sigma$ in $\Hcal$ reads
\be
\II_{\del \Sigma}^{\Hcal}(X,Y)
= \gbf(\Dbf_{X}Y, N_{\Hcal})=\frac{1}{|\nabla \dt|} \, \nabla_{\Hcal}^2 \dt(X,Y),
\ee
where $\dt$ is regarded as a function on $\Hcal$ and $\nabla_{\Hcal}$ denotes the covariant derivative
associated with the induced metric on $\Hcal$, while
$N_{\Hcal}$ is the normal vector of $\del \Sigma$ in $\Hcal$.
Similarly, we have
\be
\II_{\del \Sigma}^\St(X,Y)
= \gbf(\Dbf_{X}Y, N_\St)
= \frac{1}{|\nabla \tau|} \, \nabla_\St^2 \bar{\tau}(X,Y).
\ee

By a direct computation we find that
$$
\nabla_{\Hcal}^2 \dt(X,Y)
= \Dbf^2 \dt(X,Y) - \gbf \Big(\Dbf\dt, \frac{\del}{\del \tau}\Big) \, A(X,Y).
$$
So, we have
\be
\frac{1}{ c^2 \, C' \, r} \, \gbf (X,Y)
\leq \nabla_{\Hcal}^2 \dt(X,Y)
\leq \frac{C'}{c^2 r} \, \gbf(X,Y)
\ee
and similar inequalities for $\nabla_\St^2 \tau(X,Y)$.
On the other hand, by the triangle comparison theorem for the Riemannian metric $\gt$, there exists
a constant $C''$ (depending on $c$) such that
$$
\la N_{\Hcal}, N_\St \ra_\gbf \geq 1 + \frac{1}{C''}.
$$
This implies the existence of two functions $a, b$ that are bounded by some uniform constant $C$ and satisfy
$N_{\del \Sigma} = a \, N_{\Hcal} + b \, N_\St$. This completes the proof of Proposition~\ref{bcmc}.
\end{proof}


\section{Local coordinates ensuring the optimal regularity}
\label{FO}

\subsection{Main statements for this section}

From now on we assume that the manifold $M$ satisfies the Einstein vacuum equations. We will now prove:

\begin{theorem}[Local coordinates ensuring the optimal regularity]
\label{1kind}
Given $\eps>0$ and $q \in [1,\infty)$ there exists a constant $c(n,\eps)$
satisfying $\lim_{\eps \to 0} c(n,\eps,q) = 0$ such that the following property holds.
Let $(\Mbf, \gbf, \pbf,\Tbf_\pbf)$ be a pointed Lorentzian manifold satisfying the curvature and
injectivity radius bounds \eqref{bd3} at some scale $r>0$, together with
Einstein field equation $\Ricbf = 0$.
Then, there exists a local coordinate chart $\xbf=(x^\alpha)$ satisfying $x^\alpha(p)=0$,
defined for all $|\xbf|^2 := (x^0)^2 + (x^1)^2 + \ldots + (x^n)^2 < r_1^2$ with $r_1:=c_1(n,\eps)r$,
and such that
$$
\aligned
& \sup_{\mbox{\small \boldmath$|x|$} \leq r_1}  \Big( |\gbf_{\alpha\beta}-\etabf_{\alpha\beta}|
+ r \, |\del \gbf_{\alpha\beta}| \Big) \leq \eps,
\\
& \frac{1}{r^{n+1-2q}} \int_{\mbox{\small \boldmath$|x|$}\leq r_1} |\del^2 \gbf_{\alpha\beta}|^q \, dx
\leq C(\eps,q),
\endaligned
$$
where $\etabf_{\alpha\beta}$ is the Minkowski metric in these local coordinates.
\end{theorem}

The proof of this theorem will be given at the end of this section, after establishing several preliminary
results of independent interest.

The main observation made in the present section is that the time
function $t$ associated with the CMC foliation constructed in the
previous section admits well-controlled covariant derivatives {\sl
up to third-order;} cf.~Proposition~\ref{3t} below. Consequently,
by following our arguments given earlier in
\cite[Proposition~9.1]{ChenLeFloch} we are led to the desired
optimal regularity result in Theorem~\ref{1kind}. Recall that in
the earlier work \cite{ChenLeFloch} we relied on a coordinate
system in which the metric coefficients $\gbf$ had well-controlled
first-order derivatives only; indeed, the time function in
\cite{ChenLeFloch} was simply taken to be the geodesic distance
function, which is only twice differentiable under the curvature
and injectivity radius bounds. In contrast, in the present paper
we have constructed a more regular time function $t$ (the mean
curvature of the spacelike slices) which turns out to have
third-order regularity.

Consider the constant mean curvature foliation $\Sigma_t$ given by Theorem~\ref{foli}, and observe
that the time function $t$ (together with the Lorentzian metric $\gbf$)
provides us with a natural flow $\Phi^m$ associated with the vector field
$$
{\Dbf t \over -\gbf(\Dbf t, \Dbf t)},
$$
such that the parameter $m$ may differ from $t$ by a constant. Denote by $\tau$ the normal time function introduced
in Section~2 and recall that, in the interior of the slice,
$$
\frac{\del u}{\del \tau} < 0, \qquad
\big\la \frac{\Dbf t}{-\gbf(\Dbf t, \Dbf t)},\Dbf \tau \big\ra < 0,
$$
while the vector field $\Dbf t$ vanishes identically on the boundary.
By starting from any arbitrary point $\pbf_t \in \Sigma_t$, the integral curve $\Phi^m(\pbf_t)$ intersects
each CMC slice exactly once and the flow $\Phi^m$ preserves the CMC foliation.

Let $y=(y^i)$ be spatial coordinate chosen arbitrary on a given slice $\Sigma_{t_0}$,
and let us use the flow $\Phi^m$, in order to transport these coordinates to any other slice $\Sigma_t$.
Together with the mean curvature function $t$, these spatial coordinates provide us with spacetime coordinates
$\ybf=(t,y^i)$. The metric $\gbf$ then takes the form
\be
\label{metricform1}
\gbf = -\lambda(t,y)^2 \, dt^2 + g_{ij}(t,y) \, dy^idy^j
\ee
and satisfies the ADM equations
\be
\label{evo}
\aligned
\frac{\del g_{ij}}{\del t} & = -2 \lambda \, k_{ij},
\\
\frac{\del k_{ij}}{\del t} & = -\nabla_{i}\nabla_{j}\lambda -\lambda \, g^{pq}k_{ip}k_{qj} + \lambda \, \Rbf_{i \Nbf j \Nbf},
\endaligned
\ee
where $\lambda>0$ is the lapse function and $k_{ij}$ is the second fundamental form of $\Sigma_t$
expressed in the coordinates under consideration in this section

From now on, without loss of generality we set $r=1$. The central technical estimate of the present section
concerns the lapse function and is stated in the following lemma.

\begin{lemma}[Second-order estimates for the lapse function]
\label{propo32}
Under the assumption of Theorem~\ref{1kind} and with the above notation,
the function $\lambda$ satisfies
$$
\int_\Sigma
\left(
|\nabla^2\lambda|^2 + \Big|\frac{\del \lambda}{\del t}\Big|^2 + \Big|\frac{\del^2 \lambda}{\del t^2}\Big|^2
\right) \, dv_\Sigma \leq C(n).
$$
\end{lemma}

For any $\delta>0$ we set $\Sigma^{\delta}:=\{x\in \Sigma \, / \,  d(x,\del \Sigma)\geq \delta\}$
and, away from the boundary of the slices, we can improve Lemma~\ref{propo32}, as follows.

\begin{lemma}[Higher-order interior estimates for the lapse function]
\label{zhongyao}
Under the assumption of Theorem~\ref{1kind} and with the above notation,
for any $\delta>0$ and $q \in [1,\infty)$
one has
$$
\aligned
& \sup_{\Sigma^\delta}
\Big( |k|+|\nabla \lambda| + |\nabla^2 \lambda| + \Big|\frac{\del \lambda}{\del t}\Big|
  + \Big|\nabla\frac{\del \lambda}{\del t}\Big| + \Big|\frac{\del^2 \lambda}{\del t^2}\Big|\Big)
  \leq C(n, \delta),
\\
& \int_{\Sigma^\delta} \Big( |\nabla k|^q+|\nabla^3\lambda|^q
    + \Big|\nabla^2\frac{\del \lambda}{\del t}\Big|^q + \Big|\nabla\frac{\del^2 \lambda}{\del t^2}\Big|^q
    \Big) \, dv_\Sigma
\leq C(n,q,\delta).
\endaligned
$$
\end{lemma}

Finally, based on Lemma~\ref{zhongyao} we prove that the time
function $t$ admits well-controlled third-order derivatives. Here,
we use the covariant derivative $\Dbfhat$ associated with the
reference Riemannian metric in the coordinates under
consideration, that is,
$
\gbfhat := \lambda(t,y)^2 \, dt^2 + g_{ij}(t,y) \, dy^idy^j.
$

\begin{proposition}[Third-order estimates for the time-function]
\label{3t}
Under the assumption of Theorem~\ref{1kind} and with the above notation, for all $q \in [1,\infty)$ one has
$$
\aligned
& \sup_{\Sigma^\delta} \Big( |\Dbf^2 t| + |\Dbfhat{}^2 t| + |\Dbfhat{}^2 \lambda| \Big) \leq C(n, \delta),
\\
& \int_{\Sigma^\delta}  \Big( |\Dbf^3t|^q + |\Dbfhat{}^3 t|^q \Big) \, dv_\Sigma \leq C(n,q,\delta),
\endaligned
$$
where all the norms are computed with the reference metric $\gbfhat$.
\end{proposition}


\subsection{Derivation of the key estimates on the lapse function}

This section is devoted to the proof of Lemma~\ref{propo32}.

\noindent{\it Step 1. Zero-order estimates in time.}
By integrating Weitzenb\"ock identity \eqref{boch} and observing that $|\Delta u|+|\nabla u|\leq C$
on a CMC slice $\Sigma$, we obtain
\be
\int_\Sigma |\nabla^2 u |^2 dv_\Sigma \leq C+\int_{\del \Sigma}
(\nabla^2 u) \big(n_{\del \Sigma}, n_{\del \Sigma} \big),
\ee
where $n_{\del \Sigma} := \frac{\nabla u}{|\nabla u|}$ is the unit normal
vector field of the boundary $\del \Sigma$ on $\Sigma$. Since
\be
\label{machinery}
\aligned
(\nabla^2 u) \big( n_{\del \Sigma}, n_{\del \Sigma} \big)
& = \Delta u - tr_{g}\big( \nabla^2u\big)
\\
& = \Delta u - |\nabla u| \, tr_{g} \big( \II_{\del\Sigma}\big) \quad \text{ along } \Sigma,
\endaligned
\ee
we conclude with the boundary estimate in Proposition~\ref{bcmc} that the second fundamental form $k$
is uniformly bounded in the $L^2$ norm
\be
\label{l2k}
\int_\Sigma |k|^2 \, dv_\Sigma \leq C.
\ee
Observe that this estimate covers the whole slice up to its boundary (in contrast with the interior
sup-norm estimate given by Lemma~\ref{sec}).

We use the notation introduced in Section~2.
Recall that the Riemannian distance function $\dt = \dt(\gamma(s + s^2),\cdot)$ takes the constant value
$c_0 = s^2 + s^3$ on the boundary $\del \Sigma$,
and that $c_0 - \dt$ is proportional to the intrinsic distance function to the boundary of the slice, i.e.
\be
\label{jeiduan}
{1 \over C'} \, d(\cdot,\del \Sigma) \leq c_0 - \dt \leq C' \, d(\cdot,\del \Sigma).
\ee
Moreover, by the Laplacian comparison lemma for distance functions and relying on our curvature assumption
we have also
\be
\label{bijiao}
\frac{C'}{c^2}\geq \Delta \dt \geq \frac{1}{c^2 C'} \qquad \text{ on the hypersurface } \Sigma.
\ee

Now, taking the trace of the second identity in \eqref{evo} and recalling that $\Sigma=\Sigma_t$ has constant mean
curvature we obtain the {\sl elliptic equation satisfied by the lapse function}
\be
\label{lapse}
\Delta \lambda = -1 + \big( |k|^2 + \Ricb(N,N)\big) \, \lambda.
\ee
In view of \eqref{bijiao} and recalling that $\lambda>0$ we deduce
$$
\Delta \Big( \lambda+c^2 C' \dt \Big) \geq 0
$$
so that, thanks to the maximum principle,
\be
\label{lamda}
0 \leq \lambda \leq c^2 C' (c_0-\dt).
\ee
In particular, in view of \eqref{jeiduan} this implies the desired gradient estimate {\sl along the boundary} at least
\be
\label{bdylamda}
\sup_{\del \Sigma}|\nabla \lambda| \leq C'.
\ee

Next, by \eqref{l2k} and \eqref{lamda}, the right-hand side of \eqref{lapse} belongs to $L^2$, which
yields us a bound for the Laplacian of the lapse function
\be
\label{one}
\int_\Sigma|\Delta \lambda|^2  \, dv_\Sigma \leq C'.
\ee
On the other hand, by multiplying \eqref{lapse} by the function $\lambda$ and integrating by parts,
we find the $L^2$ gradient estimate
\be
\label{two}
\int_\Sigma |\nabla \lambda|^2 \, dv_\Sigma\leq C'.
\ee

Finally, by observing that
$$
\int_\Sigma |\Delta \lambda|^2\leq C',
\quad
\Delta \lambda\mid_{\del \Sigma}=-1,
\quad
R{ic}(\nabla\lambda,\nabla \lambda) \geq - C' \, |\nabla \lambda|^2,
$$
integrating Bochner formula
\be
\label{bolamda}
\Delta |\nabla \lambda|^2=2|\nabla^2\lambda|^2+2\la \nabla \lambda, \nabla
\Delta \lambda\ra+2R{ic}(\nabla \lambda,\nabla \lambda),
\ee
and then using \eqref{bdylamda} together with a similar calculation as in \eqref{machinery},
we arrive at an estimate of {\sl all second-order spatial derivatives} of $\lambda$:
\be
\label{2lamda} \int_\Sigma  |\nabla^2 \lambda|^2  \, dv_\Sigma \leq C',
\ee
which is one of the estimates stated in Lemma~\ref{propo32}.

Furthermore, we can also control certain nonlinear functions. Multiplying the identity
\be
\label{idi}
\Delta \lambda^2 = 2 \big( |k|^2+\Ricb(N,N) \big) \, \lambda^2 - 2 \, \lambda + 2 \, |\nabla\lambda|^2,
\ee
by $|\nabla \lambda|^2 \, \lambda^{-(1-\eps)}$ on one hand and by $|k|^2$ on the other hand,
for all $\eps \in (0,1)$ we find
\be
\label{4lamda}
\int_\Sigma \lambda^{-1+\eps} \, |\nabla\lambda|^4 \, dv_\Sigma
\leq C_\eps,
\qquad
\int_\Sigma |k|^2|\nabla \lambda|^2 \, dv_\Sigma \leq C.
\ee
Thus, by multiplying \eqref{bolamda} by $|\nabla \lambda|^2$ and then using \eqref{4lamda}, we obtain
\be
\label{5lamda}
\int_\Sigma \Big(
|\nabla^2\lambda|^2 |\nabla \lambda|^2 + |k_{ij} \, \nabla_i\lambda|^2 \, |\nabla\lambda|^2 \Big) \, dv_\Sigma
\leq C,
\ee
where we used the Gauss equation for the expression of $\Ric$. In particular, this provides us with a control
of
$$
\int_\Sigma \big| \nabla |\nabla \lambda|^4 \big| \, dv_\Sigma \leq C.
$$
Moreover, since the boundary values of $|\nabla \lambda|$ are uniformly bounded, by Sobolev inequality we also
have
$$
\int_\Sigma |\nabla \lambda|^{\frac{4n}{n-1}} \, dv_\Sigma  \leq C.
$$

\

\noindent{\it Step 2. First-order estimates in time.} This is the first instance where we use our assumption
that the manifold is Ricci-flat.
By differentiating \eqref{lapse} in time, we obtain
that
\be
\label{lapsetime0}
\aligned
\Delta \Big( \frac{\del \lambda}{\del t} \Big)
=
& \Big\la \frac{\del g_{ij}}{\del t},\nabla_i\nabla_j\lambda \Big\ra
    + \big( |k|^2 + \Ricbf(\Nbf,\Nbf) \big) \, \frac{\del \lambda}{\del t}
    + 2 \, \lambda \, \big\la \frac{\del k_{ij}}{\del t},k_{ij} \big\ra
\\
& - 2 \, \frac{\del g_{ij}}{\del t} \, k_{kl}k_{rs}g^{sl}g^{ik}g^{jr}\lambda
  + \Ricbf({\nabla}_{\frac{\del}{\del t}} \Nbf,\Nbf) \, \lambda
  + \big( {\nabla}_{\frac{\del}{\del t}}\Ricbf \big)(\Nbf,\Nbf) \, \lambda
\\
& + t \, |\nabla \lambda|^2 - 2 \, g^{il}g^{kj}k_{ij}\nabla_l \lambda\nabla_k\lambda
  + 2 \, \lambda g^{kl} \, \nabla_l\lambda \, \Rbf_{k0},
\endaligned
\ee
where we used
\be
\label{tig}
\aligned
\frac{\del}{\del
t}\Gamma^k_{ij}&=-g^{kl} \Big(
\nabla_i\lambda k_{lj}+\nabla_j\lambda
k_{li}-\nabla_l\lambda k_{ij}+\lambda(\nabla_i
k_{lj}+\nabla_jk_{li}-\nabla_l k_{ij}) \Big),
\\
-g^{ij}\frac{\del}{\del t}\Gamma^k_{ij}&=2\lambda g^{kl}
g^{ij}\nabla_ik_{lj}+2 g^{ij} g^{kl}k_{li}\nabla_j\lambda - t
g^{kl}\nabla_l \lambda\\& = -2\lambda g^{kl} \Rbf_{l0}+2
g^{ij} g^{kl}k_{li}\nabla_j\lambda - t g^{kl}\nabla_l \lambda,
\endaligned
\ee
as well as Codazzi equation
$\nabla_i k_{lj}-\nabla_{l}k_{ij}=\Rbf_{il0j}$.
Plugging in the vacuum Einstein equation $\Ricbf=0$, we arrive at the
{\sl equation satisfied by the derivative of the lapse function}
\be
\label{lapsetime1}
\aligned
\Delta \Big( \frac{\del \lambda}{\del t} \Big) - |k|^2 \, \frac{\del \lambda}{\del t}
& = 2 \, \la \Rbf_{iNjN},k_{ij}\ra \, \lambda^2 - 4 \, \lambda \la k_{ij}, \nabla_i\nabla_j\lambda \ra
\\
& \quad + 2 \, k_{ij}k_{kl}k_{rs} \, g^{sl}g^{ik}g^{jr} \lambda^2 + t \, |\nabla \lambda|^2
 - 2 \, g^{il}g^{kj}k_{ij} \, \nabla_l \lambda\nabla_k\lambda
 \\
& = - Q + \nabla_l V^l,
\endaligned
\ee
in which, thanks to our estimates in Step~1,
$$
\int_\Sigma \Big( |Q|^2+|V|^4 \big) \, dv_\Sigma \leq \Cbar,
\qquad
V^l = -2 g^{il}g^{kj}k_{ij}\lambda\nabla_k\lambda.
$$

By multiplying \eqref{lapsetime1} by $\big(\frac{\del \lambda}{\del t}\big)^{1+\eps}$ on both
sides, using Sobolev inequality for the function $\big(\frac{\del \lambda}{\del t} \big)^{1+\eps/2}$,
recalling $\frac{\del\lambda}{\del t}\mid_{\del \Sigma}=0$, and finally
integrating by parts, we conclude that 
\be
\label{lapsetime0'}
\aligned
& C'' \left(
\int_\Sigma \left| \frac{\del \lambda}{\del t} \right|^{\frac{n}{n-2}(2+\eps)} \, dv_\Sigma
 \right)^{(n-2)/n} \hskip-.5cm
\leq \int_\Sigma \left| \nabla \frac{\del \lambda}{\del t} \right|^2
          \left| \frac{\del \lambda}{\del t} \right|^{\eps} \, dv_\Sigma
       + \int_\Sigma  |k|^2 \, \left| \frac{\del \lambda}{\del t} \right|^{2+\eps} \, dv_\Sigma
\\
& \leq \left( \int_\Sigma \left| \frac{\del \lambda}{\del t} \right|^{2(1+\eps)}\right)^{1/2}
 \left( \int_\Sigma |Q|^2 \, dv_\Sigma \right)^{1/2}
\\
& \quad + \left( \int_\Sigma \left| \nabla\frac{\del \lambda}{\del t} \right|^2
      \left|\frac{\del \lambda}{\del t} \right|^{\eps} \, dv_\Sigma\right)^{1/2}
 \left(\int_\Sigma \left| \frac{\del \lambda}{\del t}\right|^{2\eps} \, dv_\Sigma \right)^{1/4}
     \left( \int_\Sigma |V|^4  \, dv_\Sigma \right)^{1/4}.
\endaligned
\ee
So, we should take $\eps\leq \frac{4}{n-4}$ if $n\geq 5$, but can take arbitrary $\eps>0$ if $n \leq 4$.
We conclude that
\be
\label{543}
\aligned
& \int_\Sigma \left| \frac{\del \lambda}{\del t} \right|^{\frac{2n}{n-4}}   \, dv_\Sigma
\leq C \qquad \text{ if } n \geq 5,
\\
& \int_\Sigma \left| \frac{\del \lambda}{\del t} \right|^q  \, dv_\Sigma \leq C_q
\qquad \, \, \, \text{ if } n =4,
\\
& \sup_\Sigma \left| \frac{\del \lambda}{\del t} \right| \leq C
 \qquad \, \, \, \, \, \qquad \text{ if } n \leq 3.
\endaligned
\ee
In the case $n\leq 3$ above, we used once more Nash-Moser's iteration technique.

Next, multiplying Bochner formula
$$
\Delta \left| \nabla \frac{\del \lambda}{\del t} \right|^2
=
\left| \nabla^2 \frac{\del \lambda}{\del t} \right|^2
 + 2 \, \la \nabla \frac{\del \lambda}{\del t}, \nabla \Delta\frac{\del \lambda}{\del t}\ra
 + 2 \, Ric\left(\nabla \frac{\del \lambda}{\del t},\nabla \frac{\del \lambda}{\del t}\right)
$$
by $(c_0-\dt)^2$, using
$$
\int_\Sigma \Big( \left| \nabla \frac{\del \lambda}{\del t} \right|^2
    + (c_0-\dt)^2 \, \left| \Delta\frac{\del \lambda}{\del t} \right|^2 \Big) \, dv_\Sigma
\leq C,
$$
and then integrating by parts, we find
\be
\label{a7}
\int_\Sigma (c_0-\dt)^2 \left| \nabla^2 \Big( \frac{\del \lambda}{\del t} \Big) \right|^2 dv_\Sigma \leq C.
\ee

Multiplying \eqref{idi}
by $|\nabla \frac{\del \lambda}{\del t}|^2$  and integrating by parts, we have
\be
\label{a8}
\aligned
\int_\Sigma 2|\nabla\lambda|^2|\nabla \frac{\del \lambda}{\del t}|^2  \, dv_\Sigma
&
\leq C
+ \int_\Sigma 4\lambda \, \left| \nabla^2 \frac{\del \lambda}{\del t} \right|
     \, \left| \nabla\frac{\del \lambda}{\del t} \right|
     \, \left| \nabla \lambda \right| \, dv_\Sigma
\\
& \leq  C + \int_\Sigma \left| \nabla\lambda\right|^2 \, \left| \nabla \frac{\del \lambda}{\del t} \right|^2 \, dv_\Sigma
    + 4 \, \int_\Sigma \lambda^2 \, \left|\nabla^2 \frac{\del \lambda}{\del t} \right|^2 \, dv_\Sigma
\endaligned
\ee
and, after using \eqref{lamda} and \eqref{a7},
\be
\label{a9}
\aligned
\int_\Sigma |\nabla\lambda|^2 \, \left| \nabla \frac{\del \lambda}{\del t} \right|^2 \, dv_\Sigma
& \leq C.
\endaligned
\ee

Now, multiplying Bochner formula \eqref{bolamda} by $|\frac{\del \lambda}{\del t}|^2$ and then integrating by parts,
we obtain
\be
\aligned
& \int_\Sigma \Big(
2 \, \left| \nabla^2\lambda \right|^2 \, \left| \frac{\del \lambda}{\del t} \right|^2
   + 2 \, R{ic}\big(\nabla \lambda,\nabla \lambda\big) \, \left| \frac{\del \lambda}{\del t} \right|^2
   \Big) \, dv_\Sigma
\\
& \leq C' \int_\Sigma \Big(
|\nabla\lambda|^2 \, \left| \nabla \frac{\del \lambda}{\del t} \right|^2
     + |\nabla^2\lambda|^2 \, \left| \frac{\del \lambda}{\del t} \right|^2 \Big) \, dv_\Sigma
  + C' \int_\Sigma (\Delta \lambda)^2 \, \left| \frac{\del \lambda}{\del t} \right|^2 \, dv_\Sigma
\endaligned
\ee
and, thanks to \eqref{a9} and \eqref{lapsetime0'},
$$ 
\aligned
\int_\Sigma \Big(
\left| \nabla^2\lambda \right|^2 \, \left| \frac{\del \lambda}{\del t} \right|^2
    + |k_{ij}\nabla_j \lambda|^2 \, \left| \frac{\del \lambda}{\del t}\right|^2 \Big) \, dv_\Sigma
\leq C.
\endaligned
$$

\

\noindent{\it Step 3. Second-order estimates in time.}
We now turn to the most involved estimate concerning the function $\frac{\del^2
\lambda}{\del t^2}$ which, we claim, satisfies an equation of the form
\be
\label{lapsetime2}
\Delta \Big( \frac{\del^2 \lambda}{\del t^2} \Big) -|k|^2\frac{\del^2 \lambda}{\del t^2}
= \nabla_iV^i+f_1+f_2+f_3,
\ee
where
$$ 
\aligned \int_\Sigma |V^i|^2 dv \leq C
\endaligned
$$
and $f_1\in L^1$ has the form $f_1 = k \ast f_1'$ (that a linear combination of such products)
with
$$
\int_\Sigma |f_1'|^2\leq C,
\qquad
\int_\Sigma |f_2|^{\frac{2n}{n+2}}\leq C,
$$
and $f_3$ is bounded pointwise by $C' \, |\nabla^2\lambda|^2$. This is one of the key
observations in the present paper.

To establish \eqref{lapsetime2} we differentiate {\eqref{lapsetime1}} in time. It is
not hard to show that all terms arising in the right-hand side of the
equation, except those of form $\nabla_iV^i$ with $\int_\Sigma |V|^2\leq \Cbar$,
belongs to $L^1$ uniformly. We emphasize that we may arrange the other terms by introducing
new terms $V^i$'s so that they all have the desired form in \eqref{lapsetime2}.

Let us now deduce from \eqref{lapsetime2} that
\be
\label{roug2}
\int_\Sigma
\left( \Big|\nabla\frac{\del^2 \lambda}{\del t^2}\Big|^2
   + \Big|\frac{\del^2 \lambda}{\del t^2}\Big|^2 \right) \, dv_\Sigma
\leq C''.
\ee
Namely, by multiplying both sides of \eqref{lapsetime2} by $\frac{\del^2 \lambda}{\del t^2}$,
then integrating by parts, and using
$\frac{\del^2 \lambda}{\del t^2}\mid_{\del\Sigma}=0$, we obtain
\be
\label{lapsetime2b}
\aligned
& \int_\Sigma \Big(
\left| \nabla\frac{\del^2 \lambda}{\del t^2} \right|^2
    + |k|^2 \, \left| \frac{\del^2 \lambda}{\del t^2} \right|^2 \Big) \, dv_\Sigma
\\
& \leq
\int_\Sigma \Big(
|V| \, |\nabla \frac{\del^2 \lambda}{\del t^2}| + \left| \frac{\del^2 \lambda}{\del t^2} \right| \, |k| \, |f_1'| \Big)
\, dv_\Sigma
\\
& \quad + \left( \int_\Sigma \left| \frac{\del^2 \lambda}{\del t^2} \right|^{\frac{2n}{n-2}} \right)^{\frac{n-2}{2n}}
\, \left( \int_\Sigma |f_2|^{2n \over n+2)} \right)^{n+2 \over 2n}
   + C \, \int_\Sigma \left| \frac{\del^2 \lambda}{\del t^2} \right| \, |\nabla^2\lambda|^2 \, dv_\Sigma.
\endaligned
\ee
Since $Ric\geq -Cg$, by multiplying Bochner formula \eqref{bolamda}) by
$|\frac{\del^2 \lambda}{\del t^2}|$ we obtain
\be
\label{ibolamda}
\aligned
2 \, \int_\Sigma |\nabla^2\lambda|^2 \, \left| \frac{\del^2 \lambda}{\del t^2} \right| \, dv_\Sigma
\leq
& \int_\Sigma  \Big(
\big| \nabla|\nabla \lambda|^2 \big| \, \left| \nabla \frac{\del^2 \lambda}{\del t^2} \right|
+ 2 \, (\Delta \lambda)^2 \, \left| \frac{\del^2 \lambda}{\del t^2} \right| \Big) \, dv_\Sigma
\\
&
+ \int_\Sigma  \Big(
|\nabla \lambda \Delta \lambda| \, \left| \nabla \frac{\del^2 \lambda}{\del t^2}\right|
 + C \, |\nabla \lambda|^2 \, \left| \frac{\del^2 \lambda}{\del t^2} \right| \Big) \, dv_\Sigma.
\endaligned
\ee

By Cauchy-Schwartz inequality and Sobolev inequality, we then
have
\be
\aligned
& \int_\Sigma \Big(
\left| \nabla\frac{\del^2 \lambda}{\del t^2} \right|^2
    + \left| \frac{\del^2 \lambda}{\del t^2} \right|^2 \Big) \, dv_\Sigma
\\
&
\leq  C \int_\Sigma |f_1'|^2 \, dv_\Sigma
  + C \, \left( \int_\Sigma |f_2|^{\frac{2n}{n+2}} \, dv_\Sigma \right)^{\frac{n+2}{n}}
     + C \int_\Sigma |\nabla^2\lambda|^2|\nabla \lambda|^2 \, dv_\Sigma
\\
& \quad
+ \int_\Sigma \Big(
\frac{(\Delta \lambda)^4}  {|k|^2}+|\nabla \lambda \Delta \lambda|^2+|\nabla\lambda|^4 \Big) \, dv_\Sigma.
\endaligned
\ee
Hence, by combining with \eqref{l2k}, \eqref{5lamda}, \eqref{lapse}, \eqref{lamda}, and \eqref{4lamda} we obtain
$$
\int_\Sigma
\Big(
|\nabla^2\lambda|^2|\nabla \lambda|^2+\frac{(\Delta \lambda)^4}
{|k|^2}+|\nabla \lambda \Delta \lambda|^2+|\nabla\lambda|^4 \Big) \, dv_\Sigma
\leq C,
$$
which gives the desired estimate \eqref{roug2}.

In summary, by combining the estimates already established in Lemma~\ref{sec} and in this proof,
we thus have
\be
\label{nec}
\aligned
& \sup_\Sigma  \lambda \, |k|
  + \int_\Sigma |k|^2 \, \Big( \left| \frac{\del \lambda}{\del t} \right|^2 + |\nabla \lambda|^2 + 1 \Big) \, dv_\Sigma
\\
&  +  \int_\Sigma  \Big(
|\nabla^2\lambda|^2 \, \big( |\nabla \lambda|^2 + \left| \frac{\del \lambda}{\del t} \right|^2 + 1 \big)
  + \left| \nabla \frac{\del\lambda}{\del t} \right|^2 \Big) \, dv_\Sigma
\\
& + \int_\Sigma  \Big( \lambda^2 \, \left| \nabla^2\frac{\del \lambda}{\del t} \right|^2
  + |\nabla\lambda|^4 + |\nabla\lambda|^2 \, |\nabla \frac{\del \lambda}{\del t}|^2
  \Big) \, dv_\Sigma
\\
& + \int_\Sigma \Big( \big| k_{ij} \nabla_j \lambda \big|^2 \, \left| \frac{\del \lambda}{\del t} \right|^2
  + \left| \frac{\del k_{ij}}{\del t} \right|^2 \, \Big) \ dv_\Sigma
\leq C.
\endaligned
\ee
In the rest of this proof, we use the notation $A\approx B$ when $A-B$ is controled by the
left-hand side of \eqref{nec}.

We can compute
\be
\label{a1}
\aligned
& \frac{\del}{\del t} \, \left(
\Delta \frac{\del \lambda}{\del t}-|k|^2\frac{\del \lambda}{\del t}
\right)
- \left( \Delta \frac{\del^2 \lambda}{\del t^2} - |k|^2 \, \frac{\del^2 \lambda}{\del t^2}
\right)
\\
& = \nabla_i \Big( 2 \, \lambda \, k_{ij}\nabla_j\frac{\del \lambda}{\del t} \Big)
    - t \, \big\la \nabla\lambda, \nabla\frac{\del \lambda}{\del t} \big\ra
    - 4 \, \lambda \, k_{ij}k_{jk}k_{ki} \, \frac{\del \lambda}{\del t}
    - 2 \, \frac{\del k_{ij}}{\del t}k_{ij} \, \frac{\del \lambda}{\del t}
\\
& \approx 2 \, \la\nabla_i\nabla_j\lambda, k_{ij}\ra \, \frac{\del \lambda}{\del t} \approx 0
\endaligned
\ee
thanks to \eqref{nec}, and
\be
\label{a2}
\aligned
\frac{\del}{\del t} \, \left( 2 \, \la \Rbf_{i \Nbf j\Nbf},k_{ij}\ra \lambda^2
\right)
& = 4 \, \la \Rbf_{i \Nbf j \Nbf},k_{ij}\ra \lambda\frac{\del \lambda}{\del t}
    + 8 \, \Rbf_{i\Nbf q\Nbf}k_{pq}k_{ip}\lambda^3
    \\
& \quad +2\la\Rbf_{i\Nbf j\Nbf},\frac{\del k_{ij}}{\del t}\ra \, \lambda^2
    + 2 \, \la\frac{\del\Rbf_{i\Nbf j\Nbf}}{\del t},k_{ij} \ra \, \lambda^2
\\
& \approx 2 \, \big\la\frac{\del\Rbf_{i\Nbf j\Nbf}}{\del t},k_{ij} \big\ra \, \lambda^2.
\endaligned
\ee
We can also compute the following:
\be
\label{a3}
\aligned
& \frac{\del}{\del t} \Big( -4 \lambda \, \la k_{ij},\nabla_i\nabla_j\lambda \ra \Big)
\\ 
& = - 4 \, \big\la k_{ij},\nabla_i\nabla_j\lambda \big\ra \, \frac{\del \lambda}{\del t}
     - 4 \lambda \, \big\la \frac{\del k_{ij}}{\del t},\nabla_i\nabla_j\lambda \big\ra
      - 4 \lambda \, \big\la k_{ij},\nabla_i\nabla_j\frac{\del\lambda}{\del t} \big\ra
\\
& \quad -8 \lambda \, k_{ij}k_{jk} \, \nabla_i\lambda\nabla_k\lambda
        + 4 \, \lambda \, |k_{ij}|^2|\nabla \lambda|^2
        + 4 \lambda^2 \, k_{ij}\nabla_ik_{kj}\nabla_k\lambda
\\
& \approx \nabla_i \big( 4 \lambda^2 \, k_{ij} k_{kj} \, \nabla_k\lambda \big)
        - 8 \lambda \, k_{ij}k_{kj} \, \nabla_i\lambda\nabla_k\lambda
        - 4 \lambda^2 \, k_{ij} k_{kj} \, \nabla_i\nabla_k\lambda,
\endaligned
\ee
\be
\label{a4}
\aligned
& \frac{\del}{\del t}\Big( 2 k_{ij}k_{kl}k_{rs}g^{sl}g^{ik}g^{jr}\lambda^2 \Big)
\\
& = 4 \, k_{ij}k_{jk}k_{ki} \, \lambda \, \frac{\del \lambda}{\del t}
   + 12 \lambda^3 \, k_{ij}k_{jk}k_{ks}k_{si}
+ 6 \, \lambda^2 \, \frac{\del k_{ij}}{\del t} k_{js}k_{si} \approx 0,
\endaligned
\ee
\be
\label{a5}
\aligned
\frac{\del}{\del t}\Big( t \, |\nabla \lambda|^2 \Big)
& = |\nabla \lambda|^2 + t 2 \lambda \, k_{ij}\nabla_i\lambda\nabla_j\lambda
    + 2 \, \nabla\Big( t \lambda \, \nabla\frac{\del \lambda}{\del t} \Big)
    - 2 t \, \lambda \Delta \frac{\del \lambda}{\del t} \approx 0,
\endaligned
\ee
and finally
\be
\label{a6}
\aligned
& \frac{\del}{\del t}\Big( -2 \, g^{il}g^{kj}k_{ij}\nabla_l \lambda\nabla_k\lambda \Big)
\\
& = - 8 \, \lambda \, k_{ij}k_{jk} \, \nabla_i\lambda \nabla_k\lambda -2\frac{\del k_{ij}}{\del t}
          \nabla_i\lambda\nabla_j\lambda
- 4 k_{ij} \, \nabla_i\frac{\del \lambda}{\del t}\nabla_j\lambda
\\
&\approx \nabla_i\big( \lambda \nabla_i |\nabla\lambda|^2 \big)
-\lambda \, \Delta|\nabla\lambda|^2 \approx 0.
\endaligned
\ee

To deal with the curvature term in \eqref{a1} we need some recall property of the curvature.
We have
$$
{\del \over \del t} \Rbf_{ikjl}
= \Dbf_0\Rbf_{ikjl} + \Gamma^\alpha_{0,i}(\gbf) \, \Rbf_{\alpha kjl}
  + \Gamma^\alpha_{0,k}(\gbf) \, \Rbf_{i\alpha jl}
  + \Gamma^\alpha_{0,j}(\gbf) \, \Rbf_{ik\alpha l}
  + \Gamma^\alpha_{0,l}(\gbf) \, \Rbf_{ikj\alpha}
$$
with
$$ 
\aligned
& \Gamma^0_{00}(\gbf)
= \frac{1}{\lambda}\frac{\del \lambda}{\del t},
    \qquad
   \Gamma^0_{0i}(\gbf) =\frac{1}{\lambda }\frac{\del \lambda}{\del y^i},
   \qquad  \Gamma^0_{ij}(\gbf) =\frac{1}{2 \lambda^2}\frac{\del g_{ij}}{\del t},
   \\
&  \Gamma^k_{00}(\gbf) =-\lambda g^{kl}\frac{\del \lambda
}{\del y^{l}},
     \qquad \Gamma^k_{i0}(\gbf) =\frac{1}{2 }g^{kl}\frac{\del g_{li} }{\del t},
     \qquad \Gamma^k_{ij}(\gbf) ={\Gamma}^k_{ij},
\endaligned
$$
$$
\aligned
\Dbf_j\Rbf_{ik0l}
= \nabla_j \Rbf_{ik0l} - \Gamma^0_{ji}(\gbf) \, \Rbf_{0k0l}
 - \Gamma^0_{jk}(\gbf) \, \Rbf_{i00l} - \Gamma^\alpha_{j0}(\gbf) \, \Rbf_{ik\alpha l}.
\endaligned
$$
Recall also the second Bianchi identity,
$$
\aligned
\Dbf_0\Rbf_{ikjl}& =-
\Dbf_j\Rbf_{ik0l}-\Dbf_l\Rbf_{ikj0}.
\endaligned
$$
We obtain
$$
\aligned
\frac{\del}{\del t}\Rbf_{i\Nbf j\Nbf}
& = - \Big(\frac{\del}{\del t} \Rbf_{ikjl} \Big) \, g^{kl}
    + \lambda \, g^{-2}\ast \textbf{Rm}\ast k
\\
& = \lambda g^{-1} {\nabla}(\Rbf_{\Nbf \ast\ast\ast})
    + \lambda \, g^{-2} \ast k \ast \textbf{Rm}
    + g^{-2} \ast \nabla \lambda \ast k \ast \Rbf_{\Nbf\ast\ast\ast}
\\
& \quad +\lambda \, \Rbf_{\Nbf\ast \Nbf\ast} \ast k \ast g^{-2}.
\endaligned
$$
So, we have
$$ 
\aligned
\frac{\del}{\del t} \Big( \la \Rbf_{iNjN},k_{ij}\ra \lambda^2 \Big) \approx 0,
\endaligned
$$
which completes the proof of Lemma~\ref{propo32}.


\subsection{Proofs of the main statements}

\begin{proof}[Proof of Lemma~\ref{zhongyao}]
Since the second fundamental form is bounded in each slice $\Sigma^\delta$, then according to Gauss equation
the intrinsic curvature of $\Sigma^\delta$ is also uniformly bounded by
$C \, \delta^{-2}$.
Hence, from the injectivity radius theorem of Cheeger, Gromov, and Taylor
for Riemannian manifolds \cite{CheegerGromovTaylor}, it follows that
the injectivity radius of $\Sigma^\delta$ is uniformly bounded from below
by $C \, \delta$. Next, using the theorem of Jost and Karcher \cite{JostKarcher},
we can find a fixed number of harmonic coordinate charts covering $\Sigma^\delta$ and
in which the metric is equivalent to the Euclidean metric and has $W^{2,q}$ regularity for
each $q\in [1,\infty)$.
In addition, by Sobolev's embedding theorem, the metric coefficients also belong to the H\"older space
$C^{1,\alpha}$ for all $\alpha \in (0,1)$.

Next, using an $L^q$ estimate from the equation \eqref{simon1} satisfied by the second fundamental form
in these harmonic coordinates, we deduce that $k_{ij} \in W^{1,q}$ for all $q$.
We also observe that the Christoffel symbols are of class $C^\alpha$, so that this
also provides us that $\nabla k \in L^q$. All implied constants are uniform and only depend on the dimension
$n$ and the distance $\delta$ to the boundary of the slice. 

Then, using a standard $W^{2,q}$ regularity estimate for equation \eqref{lapse} satisfied by the lapse function
(see, for instance, \cite{GilbargTrudinger})
and noting that
$g \in W^{2,p}$, we deduce that $\del^3\lambda \in L^q$.
Here,
$\del^3\lambda$ stands for any natural derivative in the constructed harmonic coordinates.
Since $\del \Gamma \in L^q$, we obtain the third-order covariant derivatives $|\nabla^3\lambda|\in L^q$.
Finally, we emphasize that for the spatial regularity of
$\frac{\del \lambda}{\del t}$ and $\frac{\del^2\lambda}{\del t^2}$,  we need the crucial estimates
established earlier in Lemma~\ref{propo32}.  We use an $W^{2,q}$ regularity estimate to equations
\eqref{lapsetime1} and use Lemma~\ref{propo32}; this leads to the desired estimate for
$\frac{\del \lambda}{\del t}$. Finally, the above estimates imply
$\Delta \frac{\del^2 \lambda}{\del t^2}-|k|^2\frac{\del^2 \lambda}{\del t^2}\in W^{-1,q}$,
and we use again Lemma~\ref{propo32} and an $L^p$ regularity estimate in order to control $\frac{\del^2
\lambda}{\del t^2}$.
\end{proof}


\begin{proof}[Proof of Proposition~\ref{3t}]
We now want to control the covariant derivatives of the function $t$.
Since this question is independent of the choice of coordinates, then
on the fixed slice $\Sigma^\delta$ we choose finitely many spatially harmonic coordinates patches
as in the previous proof, and we use them our new coordinates $y^i$. Then, on this fixed
time slice, the spatial metric belongs to $W^{2,q}$ for all $q \in [1,\infty)$.
In particular, the Christoffel symbols $\Gamma(\gbf)=(\Gamma^k_{ij})$ are uniformly bounded and
$\del_y \Gamma(g)\in L^q$. Combining \eqref{metricform1} and Lemma~\ref{zhongyao} together,
we see that $\Gamma(\gbf)$ is bounded and $|\del_{t,y} \Gamma(\gbf)|\in L^q$ at
slice $\Sigma^\delta$ in these particular harmonic coordinates and at this fixed time.

Now, we calculate the covariant derivatives of $t$ in the coordinates chosen above. Since
$\Dbf_{\alpha\beta}^2t = -{\Gamma(\gbf)}_{\alpha\beta}^0$, we
have $\sup_{\Sigma^\delta} |\Dbf^2t| \leq C$. For the third-order
derivative, we write
$
\Dbf^3t = \del_{t,y} \Gamma(\gbf) + \Gamma(\gbf) \ast \Gamma(\gbf) \in L^q.
$
A direct computation yields us
$$
\aligned
& \Gamma^0_{00}(\gbf)- \Gamma^0_{00}(\gbfhat)=0,
&& \Gamma^0_{0i}(\gbf)- \Gamma^0_{0i}(\gbfhat)=0,
&&&  \Gamma^0_{ij}(\gbf)-\Gamma^0_{ij}(\gbfhat) =-\frac{2}{\lambda}k_{ij},
\\
& \Gammat^k_{00} = 2 \, \lambda g^{kl} \, \frac{\del \lambda}{\del x^{l}},
&& \Gamma^k_{i0}(\gbf)- \Gamma^k_{i0}(\gbfhat) = 0,
&&&  \Gamma^k_{ij}(\gbf)-{\Gamma}^k_{ij}(\gbfhat)=0.
\endaligned
$$
Since $\Dbfhat{}^2 t - \Dbf^2 t = - \frac{2}{\lambda}k$, we have
$\sup_{\Sigma^\delta} |\Dbfhat{}^2 t| \leq C_\delta$,
and
\be
\label{zhongjian}
\aligned
\Dbf_l k_{ij} & = \nabla_{l}k_{ij}, \qquad
\Dbf_0 k_{ij} = -\nabla_i\nabla_j\lambda + \lambda g^{lp} k_{pj} k_{li} + \lambda \, \Rbf_{iNjN},
\\
\Dbf^2_{ij}\lambda & =\nabla_i\nabla_j\lambda - \frac{1}{\lambda} \, k,
\qquad
\Dbf^2_{0j}\lambda = \nabla_i \frac{\del \lambda}{\del t} + \lambda \, k \ast \nabla\lambda + \lambda^{-1} \frac{\del \lambda}{\del t}\nabla\lambda,
\\
\Dbf^2_{00}\lambda
& = \frac{\del^2 \lambda}{\del t^2} - \frac{1}{\lambda} \left( \frac{\del \lambda}{\del t} \right)^2
        - \lambda \, |\nabla \lambda|^2.
\endaligned
\ee
Therefore, we have
$$
\aligned
\Dbfhat^3 t
& =( \Dbfhat - \Dbf) \Dbfhat {}^2 t + \Dbf(\Dbfhat {}^2 t - \Dbf^2t ) + \Dbf^3t
\\
& = \Big( \frac{1}{\lambda}k + \frac{1}{\lambda}\Dbf\lambda \Big) \ast
      \Big(\frac{1}{\lambda} \, k + \frac{1}{\lambda} \Dbf\lambda \Big)
      + \Dbf\Big( \frac{k}{\lambda} \Big) + \Dbf^3t,
\endaligned
$$
and the result follows from (\ref{zhongjian}) and Lemma~\ref{zhongyao}.
\end{proof}


\begin{proof}[Proof of Theorem \ref{1kind}]
By a direct computation (see for instance \cite{ChenLeFloch}) one can check that
the Riemannian curvature of the metric $\gbfhat$ on $\bigcup_t \Sigma_t$ is uniformly bounded
and, actually,
$$
\sup_{\Sigma^\delta} |\widehat{R}_{\alpha\beta\gamma\delta}-\Rbf_{\alpha\beta\gamma\delta}|\leq
C_{\delta} \, \big( |\nabla^2 \lambda|+ |k|^2+|\Dbf\lambda|^2 \big),
$$
hence
$
\sup_{\Sigma^{\delta}} |\widehat{R}_{\alpha\beta\gamma\delta}| \leq C_\delta'.
$
According to the injectivity estimate for Riemannian manifolds established in \cite{CheegerGromovTaylor},
the injectivity radius of the metric $\gbfhat$ at the point $\pbf$ is uniformly bounded from below,
i.e.
$\text{inj}(\Mbf, \gbfhat, \pbf) \geq c_\delta'$. Therefore, according to Jost and Karcher
\cite{JostKarcher}, we may choose harmonic coordinates
$\xbf^\alpha$ of $\gbfhat$ around $\pbf$.

Noting that
$
\gbfhat = \gbf + 2\lambda^2 dt\otimes dt,
$
we obtain
$$
\aligned
\Dbfhat{}^2\gbf
= & 2\, \Dbfhat{}^2\lambda^2 \otimes \Dbfhat t \otimes \Dbfhat{} t
    + 2 \, \lambda^2 \, \Dbfhat{}^3 t \otimes \Dbfhat t
   \\
   & + 2 \, \lambda^2 \, \Dbfhat{} t \otimes \Dbfhat{}^3 t
    + 4 \, \lambda^2 \, \Dbfhat{}^2 t \otimes \Dbfhat{}^2 t.
\endaligned
$$
By combining with the result \eqref{3t}, we see that the coefficients of $\gbf$ in harmonic coordinates
$\xbf^\alpha$ for the Riemannian metric $\gbfhat$ belongs to $W^{2,q}$.
\end{proof}


\section{CMC--harmonic coordinates of an observer}
\label{CMCh}

\subsection{Construction of local coordinates}

\subsubsection*{Preliminaries} In Theorem \ref{1kind}, we constructed coordinates in which the
Lorentzian metric coefficients have optimal regularity. However, one inconvenient
of these coordinates is that they are not consistent with the CMC foliation constructed in
Section~2. In the present section, we show that both strategies can be combined and
we construct a new coordinate system which is based on the CMC foliation and
has optimal regularity as stated in Theorem \ref{1kind}. The basic idea is to now choose spatial
harmonic coordinates on {\sl each} CMC hypersurface. This strategy goes back to
Anderson~\cite{Anderson-regularity} who, however, used the time function given by a distance function.
In contrast, in our present construction, the time function coincides with the mean curvature function of
CMC slices and has much better regularity.

In view of Theorem~\ref{foli}, since the second fundamental form is bounded in each slice
$\Sigma_t$ and the spacetime curvature is bounded,
Gauss equation implies that the intrinsic curvature of the slice is also bounded.
So, according to \cite{CheegerGromovTaylor}
there exists a constant $\eta=\eta(n)>0$
so that the injectivity radius of the slice is bounded below, that is, 
$
\inj(\Sigma_t,p_t)\geq 2 \eta \, r,
$
where $p_t$ is the orbit of the base point $\pbf$ along the above flow. By a theorem established by
Jost and Karcher for Riemannian manifolds \cite{JostKarcher}, there exists a constant $\eta'=\eta'(n)>0$
such that a harmonic
coordinate system $|y|\leq \eta' r$ exists around $\pbf$ on the slice $\Sigma_{t(p)}$,
 with $\pbf = (0,\ldots,0)$ and, on that slice,
\be
{1 \over 2} \, \delta_{ij} \leq g_{ij}
= g\Big(\frac{\del}{\del y^i}, \frac{\del}{\del y^j}\Big)
\leq 2 \, \delta_{ij}.
\ee
By using the above mentioned flow, the coordinate functions $y^i$ can be extended to other slices
$\Sigma_t$ and, together with the time function $t$,
yield a spacetime coordinate system. Then, the Lorentzian metric $\gbf$ takes the form
$
\gbf = -\lambda(t,y)^2 \, dt^2 + g_{ij}(t,y) \, dy^idy^j.
$
From the estimate of $|\Dbf t|^2$ given by Theorem~\ref{foli} and in view of
the expression $\nu(\nabla u) \, \frac{\del u}{\del t}=-\lambda$, we deduce that
\be
\frac{\sqrt{\theta}}{r^2} \leq  \lambda \leq \frac{1}{r^2\sqrt{\theta}}.
\ee
Moreover, in view of the results in \cite{JostKarcher} and
thanks to \eqref{ub} and \eqref{simon1}, 
we have the uniform control
$$
\nabla K, {\nabla}^2 \lambda  \in L^q (\Sigma_t), \qquad q \in [1,\infty).
$$


\subsubsection*{Almost linear coordinates}

We now construct the coordinates of interest in this section. We can assume $r=1$.
For each $i=1, \ldots, n$ and for {\sl each slice} $\Sigma_t$
let $x^i$ be the solution of the Dirichlet problem
\be
\label{dirichlet1}
\aligned
\Delta_t
x^i & =0 \ \ \ \ \text{in } \Sigma_t\cap\{y: |y|<\eta'\},
\\
x^i & =y^i \ \  \ \ \text{ on } |y| =\eta'.
\endaligned
\ee
Let $\tbar$ be such that the slice $\pbf \in \Sigma_\tbar$ has mean curvature $\bar{t}$.

By applying the maximum principle for the operator $\Delta$ we can derive some basic properties
of the above functions. First of all, at the time~$\tbar$ one has
\be
\label{bla0}
n \leq \Delta_{\bar{t}} |y|^2=2\sum_{i=1}^n g^{ii} \leq 4n \qquad \text{ on the slice } \Sigma_\tbar,
\ee
where we have solely used that the metric coefficients are uniformly bounded.
Since
$$
\left| {\del \over \del t} \Big( \Delta_t |y|^2 \Big) \right|
=
\big| \nabla \lambda \ast k \ast \nabla |y|^2 \big|
\leq C_\delta \qquad \text{ on the slice } \Sigma_\tbar,
$$
and $\nabla \lambda$ and $k$ are uniformly bounded,
we deduce from \eqref{bla0} that
$$
{n \over 2} \leq \Delta_t |y|^2 \leq 8 n \quad \text{ on any slice }  \Sigma_t \cap \{|y|\leq \eta'\}
$$
for all $|t-\bar{t}|\leq \frac{1}{C_{\delta}C'' n}$. From now on we drop the subscript $t$ in the notation.

Then, by the same arguments as the ones above we find
\be
\label{bla}
|\Delta y^i| \leq \eps \quad \text{ on } \Sigma_t \cap \{|y|\leq \eta'\}
\ee
for all $|t - \tbar |\leq \frac{\eps}{C_\delta}$.
Now,
since $\Delta (x^i-y^i)=- \Delta y^i$, by the maximum principle we obtain
$$
C'' \eps \eta'(\eta'-|y|) \geq x^i-y^i\geq - C'' \eps \eta'
(\eta'-|y|)
$$
on the slice $\Sigma_{t}\cap \{|y|\leq \eta'\}$ for all $|t-\bar{t}|\leq \frac{\eps}{C_{\delta}C'' n}$.
In particular, along the boundary $\big\{ |y|=\eta \big\}$
the above property implies
$$
\sup_{|y|=\eta'}|\nabla (x^i-y^i)|\leq C(n)\eps \eta'
\qquad \text{ for } |t-\bar{t}|\leq \frac{\eps}{C_{\delta}C'' n}.
$$

Next, we can also estimate
$\sup_{|y|\leq \eta'}|\nabla (x^i-y^i)|$  from the equation satisfied by the coordinates,
as follows. By integration by parts we obtain
\be
\label{bla3}
\int_{|y|\leq \eta'} |\nabla (x^i-y^i)|^2 \leq C(n) \, \eps (\eta')^{n+1}.
\ee
Second, let $w= \max\big( 0, |\nabla (x^i-y^i)|^2-C(n)\eps \eta'\big)$ and consider Bochner formula
$$
\aligned
& \Delta |\nabla (x^i-y^i)|^2
\\
& =2|\nabla^2 (x^i-y^i)|^2+2\la
\nabla(x^i-y^i),-\nabla \Delta y^i\ra+2Ric(\nabla
(x^i-y^i),\nabla (x^i-y^i)),
\endaligned
$$
multiply it by $w^a$ for $a>0$,
and integrate by parts. Then, using \eqref{bla} together with Sobolev inequality and Nash-Moser technique,
we arrive at the sup-norm gradient estimate
\be
\label{thzb}
\sup_{|y|\leq \eta'} w
\leq \frac{C}{\eta'^n} \int_{|y|\leq \eta'} w \, dy \leq C(n) \eta'\eps
\ee
for all $|t-\bar{t}| \leq \frac{\eps}{C_{\delta} C'' n}$. The latter inequality follows from \eqref{bla3}.

By choosing $\eps$ suitably small (depending upon the dimension only), \eqref{thzb}
implies that the harmonic map $\Psi=(x^1,\ldots, x^n)$ is a local diffeomorphism from $\{|y|\leq \eta'\}\cap \Sigma_t$  onto its image.
By the maximum principle, $\Psi$ is a map from $\{|y|\leq \eta'\} \cap \Sigma_t$
to $\{|x|\leq \eta' \}$, which leaves invariant the boundary.
Hence, $\Psi$ is a diffeomorphism from $\{|y|\leq \eta'\}\cap \Sigma_t$ to the Euclidean ball $\{|x|\leq \eta' \}$,
and  $x=(x^1,\ldots, x^n)$ with $\{|x|\leq \eta' \}$ is a harmonic coordinate system.
Moreover, by choosing $\eps$ sufficiently small in \eqref{thzb} we find
$$
{1 \over 4} \, \delta_{ij}
\leq
{1 \over 2} \, g\Big( {\del \over \del y^i}, {\del \over \del y^i} \Big)
\leq
g\Big( {\del \over \del x^i}, {\del \over \del x^j} \Big)
\leq 2 \, g\Big( {\del \over \del y^i}, {\del \over \del y^i} \Big)
\leq 4 \, \delta_{ij},
$$
as inequalities between symmetric tensors.


\subsubsection*{The ADM formulation}

Including also $x^0=t$ in the coordinates, we have therefore constructed local spacetime coordinates
$(x^0,x^1,\ldots,x^{n})$ covering a neighborhood of the point $\pbf$. Recall that $\Nbf$ denotes the unit normal
vector along slices $\Sigma_t$. 

Note that the function $t$ appears as the time coordinate in two different
coordinate systems,
that is,
$\Psi'=(y^0,y^1,\ldots, y^{n})$ and $\Psi=(x^0,x^1,\ldots, x^{n})$ with $y^0=x^0=t$.
It is easy to see that
$$
\lambda N=\Psi'^{-1}_{\ast} \left( \frac{\del}{\del y^0} \right)
= \frac{\del}{\del x^0} + \frac{\del x^i}{\del t}\frac{\del }{\del x^i}
= \frac{\del}{\del x^0} - \xi,
$$
and so
$
\frac{\del}{\del x^0}=\lambda N+\xi,
$
where we refer to $\xi=\sum_{i=1}^{n}\frac{\del x^i}{\del t}\frac{\del}{\del
x^i}=\sum_{i=1}^{n}\xi^i\frac{\del}{\del x^i}$ as the shift vector. Define $A_{ij}=(\frac{\del x^i}{\del
y^j})$, and $\widetilde{g}_{kl}= g_{ij}(y,t){A^{-1}}_{ik}{A^{-1}}_{jl}$.
It is not hard to see that $ dy^0=dx^0 $ and $dy^i = {A^{-1}}_{ik} \Big( dx^k - \frac{\del x^k}{\del y^0}dx^0 \Big)$,
hence in the coordinates $(x^0,x^1,\ldots, x^{n})$ the metric $\gbf$
has the form
\be
\label{metricform2}
\aligned
\gbf & = - \lambda(t,y)^2 \, dt^2 + g_{ij}(t,y) \, dy^i dy^j
\\
& = - \lambda^2 \, (dx^0)^2 + \widetilde{g}_{ij}(x^0, x) \, \big( dx^i+\xi^idx^0\big) \big( dx^j+\xi^jdx^0 \big).
\endaligned
\ee
For simplicity in the notation, we drop the tilde from
$\widetilde{g}_{ij}$ and simply write the {\sl metric decomposition} as
\be
\label{metricform3}
\gbf
=-\lambda^2 (dx^0)^2+{{g}}_{ij}(x^0, x)\big( dx^i+\xi^idx^0\big) \big(dx^j + \xi^jdx^0\big).
\ee

Recall that the second fundamental form is defined by $k_{ij}=\la
\Dbf_{\frac{\del}{\del x^i}}\frac{\del}{\del
x^j},N\ra$, where $\Dbf$ is the covariant derivative
associated with the metric $\gbf$,  and recall Gauss-Codazzi
equations
\be
\label{gausscodazzi}
\aligned
& \Rbf_{ijkl} = R^{\Sigma}_{ijkl}+ k_{ik}k_{jl}-k_{il}k_{kj},
\\
& \nabla_{l}k_{ij}-\nabla_{i}k_{lj} =\Rbf_{liNj}.
\endaligned
\ee
The geometry of the slice is determined by the induced metric
$g_{ij}$ and the second fundamental form $k_{ij}$, both, satisfying the
following evolution equations:
\be
\label{evolution}
\aligned
\frac{\del g_{ij}}{\del x^0} & = -2 \lambda k_{ij} + \Lcal_{\xi}g_{ij},
\\
\frac{\del k_{ij}}{\del x^0} & = -\nabla_i \nabla_j\lambda
        + \Lcal_\xi k_{ij} - \lambda \, g^{pq} k_{ip} k_{qj} + \lambda \, \Rbf_{iNjN}. 
\endaligned
\ee
Note also that since $x^1,\ldots,x^n$ are harmonic coordinates on $\Sigma_t$, we have
\be
\label{metricequation}
g^{kl}\frac{\del^2}{\del x^k\del x^{l}}g_{ij}+Q(\del g,\del
g)=-2R_{ij},
\ee
where $Q_{ij}(\del g,\del g)$ is some quadratic expression in $\del g$ with coefficients depending on the inverse
metric $g^{-1}$.


\subsubsection*{Estimating the shift vector}

Next, we derive the equation for the shift vector $\xi$. By
differentiating the harmonic equation $\Delta x^k=0$ with respect to $x^0$, and using (\ref{evolution}), we get
\begin{equation*}
\begin{split}
0 & =  g^{kl}g^{ij} \, \Big( \nabla_i(-2\lambda k_{lj}+\nabla_j\xi_l+\nabla_l\xi_j) +\nabla_j(-2\lambda
       k_{li}+\nabla_i\xi_l+\nabla_l\xi_i)
\\
& \quad -\nabla_l(-2\lambda k_{ij}+\nabla_j\xi_i+\nabla_i\xi_j) \Big)
\\
& = 2 \, \Big( \Delta \xi^k+g^{ki}R_{ij}\xi^j+ g^{kl}\nabla_{l}(\lambda tr K)
   - 2g^{kl}g^{ij}k_{li}\nabla_{j}\lambda
- 2 \lambda (trk)_l+2\lambda g^{kl}\Rbf_{lN} \Big),
\end{split}
\end{equation*}
where $\Delta \xi^k$ is the $k$-th component of $\Delta \xi$.
By combining this result with the constant mean curvature equation,
this gives us the {\sl elliptic equation satisfied by the shift vector}
\be
\label{xiequation}
\Delta \xi^k
= - g^{ki}R_{ij}\xi^j - (tr k) g^{kl} \, \nabla_l\lambda
   + 2 g^{kl}g^{ij}k_{li} \nabla_j \lambda
  - 2\lambda \, g^{kl}\Rbf_{lN}.
\ee
It is easy to see
$
\Delta |\xi|\geq -C
$
for some constant $C$ depending only on the dimension. By choosing a sufficiently
large constant we obtain
$
\Delta \big( |\xi| + C' \, |x|^2 \big) \geq 0,
$
hence by the maximum principle we arrive at the following sup norm estimate for the shift vector
\be
\label{estimatexi}
|\xi| \leq C(n)(\eta'-|x|).
\ee


\subsection{Proof of the main theorem}

We are now in a position to give the proof of Theorem \ref{main}.  By scaling, we may assume $r=1$.

\noindent{\it Step 1. Spatial derivative estimate.}
We are going to use (\ref{simon1}) (\ref{metricequation}),  (\ref{lapse}), (\ref{xiequation}),
together with elliptic regularity estimates, and establish a bound for the spatial derivatives of the metric $\gbf$.

By choosing other harmonic coordinates on each slice, letting
$\eta'$ suitably small, and recalling the $L^p$ regularity estimates for uniformly elliptic operators,
we find for all $q \in [1, \infty)$
\be
\int_{|x|\leq \eta'} \left| \nabla^2(x^i-y^i) \right|^q \leq C_q.
\ee 
This implies that
\be
\label{1sm}
\int_{|x|\leq \eta'} \left| \frac{\del g_{ij}}{\del x^k} \right|^q
\leq C_q
\ee
and, therefore, for all $\alpha \in (0,1)$, 
$
\| g \|_{C^{\alpha}\{|x|\leq \eta'\}} \leq C_\alpha.
$
In view of \eqref{metricequation}, we obtain
\be
\int_{|x|\leq \eta'} \left| g^{kl}\frac{\del^2}{\del x^k\del x^{l}} \Big( (\eta'^2-|x|^2) \, g_{ij} \Big) \right|^q
\leq C_q
\ee
and using $L^p$ estimate, since the coefficient of the Laplacian operator are H\"older continuous
and the function under consideration vanishes on the boundary,
\be
\label{2smetric}
\aligned
\left| \frac{\del}{\del x^k}g_{ij} \right|
& \leq \frac{C(n)}{\eta'-|x|},
\\
\int_{|x|\leq \eta'} \left|\frac{\del^2}{\del x^k\del x^l} \Big( (\eta'^2-|x|^2)g_{ij} \Big) \right |^q
& \leq C_q.
\endaligned
\ee

Note that in the expression
$
\Delta \xi=g^{kl}\frac{\del^2 \xi}{\del x^k \del x^l}+\Gamma \ast \nabla \xi + \del \Gamma \ast \xi,
$
we have $\Gamma \in L^q$ (thanks to \eqref{1sm}) and
$$
|\del \Gamma \ast \xi|\leq C \, (|\del^2g|+|\del g|^2) (\eta'^2-|x|^2)
$$
thanks to \eqref{estimatexi}. The latter term belongs to $L^q$ in view of \eqref{2smetric} and therefore
$L^p$ regularity estimates applied to \eqref{xiequation} yield
$$
\sup |\del_x \xi| + \int_{|x|\leq \eta'}|\del_x^2\xi^k|^q\leq C_q
$$
or, in covariant form, we have estimated the first- and second-order derivatives of the shift vector
\be
\label{xis2c}
\sup |\nabla\xi| +\int_{|x|\leq \eta'} |\nabla^2\xi|^q\leq C_q.
\ee

In addition, since $\del_x k=\nabla k+\Gamma\ast k$ and $|\nabla k| \in L^q$ by Lemma~\ref{zhongyao},
we also find
$$
\int_{|x|\leq \eta'}|\del_x k|^q \leq C_q.
$$
Similarly, since $\del_x^2\lambda=\nabla^2\lambda+\Gamma \ast \nabla \lambda$
and in view of Lemma~\ref{zhongyao}, we also obtain
$$
\int_{|x| \leq \eta'}|\del_x^2 \lambda|^q \leq C_q.
$$
In summary, we have now control the spatial derivatives (up to second order) of the metric,
the lapse function, and the shift vector:
\be
\label{spacial}
(\eta'^2-|x|^2) g_{ij}, \, \lambda, \xi^i \in W^{2,q}\big(\{|x|\leq \eta'\}\big).
\ee

\

\noindent{\it Step 2. Estimates of first-order time derivatives.}
The strategy now is to differentiate the equations \eqref{simon1}, \eqref{metricequation},
\eqref{lapse}, and \eqref{xiequation} with respect to $t$ and then use the elliptic regularity property.

First of all, thanks to Step 1 we have
$
\Lcal_\xi g_{ij} = \nabla_i \xi_j  +\nabla_j \xi_i \in W^{1,q}_x,
\qquad
\lambda \, k_{ij} \in W^{1,q}_x
$
for all $q \in [1,\infty)$. By \eqref{evolution} we have
$\frac{\del g_{ij}}{\del {x^0}}\in W^{1,q}_x$ and , in particular, $\frac{\del^2 g_{ij}}{\del x\del {x^0}}\in L^q_x$, i.e.
in other words for all $q \in [1,\infty)$
\be
\sup_{|x|\leq \eta'} \left| \frac{\del g_{ij}}{\del x^0} \right|
+
\int_{|x|\leq \eta'} \left| \frac{\del^2 g_{ij}}{\del x\del {x^0}} \right|^q
\leq C_q.
\ee

In view of Step~1 and \eqref{evolution} again, we have $\frac{\del k_{ij}}{\del {x^0}}\in L_x^q$
     for all $q\in [1,\infty)$. Then, from Lemma~\ref{zhongyao} we deduce
\be
\label{nbgji1}
\aligned
& \sup_{|x|\leq \eta'}
\left( \left| \frac{\del \lambda}{\del t} \right| + \left| \nabla\frac{\del \lambda}{\del t} \right|
 + \left| \frac{\del^2 \lambda}{\del t^2} \right| + \left| \nabla^2\lambda \right| \right)
\\
& + \int_{|x|\leq \eta'} \left(  \left| \nabla k \right|^q + \left| \nabla^3\lambda \right|^q
  + \left| \nabla^2\frac{\del \lambda}{\del t}\right|^q
  + \left| \nabla\frac{\del^2 \lambda}{\del t^2} \right|^q
  \right)
\leq C_q.
\endaligned
\ee

Since $\frac{\del \lambda}{\del x^0} = \frac{\del \lambda}{\del t} + \la\xi, \nabla\lambda\ra$,
we have
$$
\nabla \frac{\del \lambda}{\del x^0}
= \nabla \frac{\del \lambda}{\del t}+\nabla \xi\ast
\nabla\lambda+\xi \ast \nabla^2\lambda
$$
and
$$
\nabla^2 \frac{\del \lambda}{\del x^0}= \nabla^2 \frac{\del
\lambda}{\del t}+\nabla^2 \xi\ast \nabla\lambda+\xi \ast
\nabla^3\lambda +\nabla\xi\ast \nabla^2\lambda.
$$
Then by \eqref{nbgji1}, \eqref{xis2c}, and Step 1, we find
\be
\sup_{|x|\leq \eta'}|\frac{\del \lambda}{\del
x^0}|+|\nabla \frac{\del \lambda}{\del x^0}|+\int_{|x|\leq
\eta'}|\nabla^2\frac{\del \lambda}{\del x^0}|^q\leq C_q.
\ee

Next, note that
$$
\frac{\del}{\del {x^0}}\nabla_i\nabla_j\xi^k=\nabla_i\nabla_j\frac{\del
\xi^k}{\del {x^0}}+\nabla(\frac{\del g}{\del {x^0}})\ast
g^{-1}\ast \nabla \xi+g^{-1}\ast \nabla^2(\frac{\del g}{\del {x^0}})\ast \xi.
$$
By differentiating (\ref{xiequation}) with respect to the
time variable $x^0$, we obtain
\be
\label{xiequationtime0}
\Delta
\frac{\del \xi^k}{\del {x^0}} = A_1+A_2+A_3,
\ee
with
$$
\aligned
A_1 :=
& -g^{ki}R_{ij}\frac{\del \xi^j}{\del {x^0}}-tr k g^{kl}\nabla_{l}\frac{\del \lambda}{\del
{x^0}}+2g^{kl}g^{ij}k_{li}\nabla_{j}\frac{\del \lambda}{\del
{x^0}}+2g^{kl}g^{ij}\frac{\del k_{li}}{\del {x^0}}\nabla_{j}\lambda,
\endaligned
$$
$$
\aligned
A_2 :=
& -\frac{\del R_{ij}}{\del {x^0}}g^{ki}\xi^j-2\lambda g^{kl}(\nabla_{\frac{\del}{\del
{x^0}}}\Ricbf)(\frac{\del}{\del x^{l}}, N)-2 \lambda
g^{kl}\Ricbf(\nabla_{\frac{\del}{\del {x^0}}}\frac{\del}{\del x^{l}}, N)
\\
& -2 \lambda g^{kl}\Ricbf(\frac{\del}{\del x^{l}},
\nabla_{\frac{\del}{\del {x^0}}}N)  +\frac{\del g_{rs}}{\del
{x^0}} \Big( g^{kr}g^{si}R_{ij}\xi^j + tr k \, g^{kr} g^{ls} \nabla_l \lambda
\\
& -2g^{kr}g^{sl}g^{ij}k_{li}\nabla_{j}\lambda-2g^{kl}g^{ir}g^{js}k_{li}\nabla_{j}\lambda
+2\lambda g^{kr}g^{ls}\Rbf_{lN} + g^{ir}g^{js}\nabla_i\nabla_j\xi^k \Big),
\endaligned
$$
and
$$
\aligned
A_3 := \nabla\left( \frac{\del g}{\del {x^0}} \right) \ast g^{-2}\ast \nabla \xi
       + g^{-2} \ast \nabla^2\left( \frac{\del g}{\del {x^0}} \right) \ast \xi.
\endaligned
$$

Since the spacetime under consideration satisfies the vacuum Einstein equation, we obtain
\be
\label{xiequationtime0-2}
\aligned
 \Delta
\frac{\del \xi^k}{\del {x^0}}
=
& - g^{ki}R_{ij}\frac{\del
\xi^j}{\del {x^0}}+g^{-2}\ast \Big( k \ast \nabla \frac{\del \lambda}{\del {x^0}}+\frac{\del k}{\del {x^0}}\ast\nabla \lambda\Big)
+ g^{-2}\ast \nabla^2(\frac{\del g}{\del {x^0}})\ast \xi\\
& +\frac{\del Ric}{\del {x^0}}\ast g^{-1}\ast
\xi+\nabla(\frac{\del g}{\del {x^0}})\ast g^{-2}\ast \nabla
\xi
\\
& + \frac{\del g}{\del {x^0}}\ast \, \Big(
g^{-2}\ast Ric\ast \xi+ k\ast g^{-3}\ast\nabla \lambda+g^{-2}\ast\nabla^2\xi \Big).
\endaligned
 \ee
It is a classical observation that
\be
\label{metrictime1}
\aligned
-2 \, {\del R_{ij} \over \del x^0}
=
\Delta_L\left( {\del g_{ij} \over \del x^0} \right) + \nabla_i V_j + \nabla_j V_i,
\endaligned
\ee
where
$$
\Delta_{L}(\frac{\del g_{ij}}{\del {x^0}})=\Delta (\frac{\del g_{ij}}{\del {x^0}})
+2R_{ikjl}\frac{\del g_{kl}}{\del {x^0}}-R_{ik}\frac{\del g_{kj}}{\del {x^0}} -R_{jk}\frac{\del g_{ki}}{\del
{x^0}}
$$
is the Lichnerowicz Laplacian, and
$$
V_i := {1 \over 2}\nabla_i (g^{kl}\frac{\del g_{kl}}{\del {x^0}})
       - g^{kl}\nabla_k \frac{\del g_{il}}{\del {x^0}}.
$$
Since $\nabla \frac{\del g}{\del {x^0}}\in L^q_x, \nabla^2 \frac{\del g}{\del
{x^0}}=\del \nabla\frac{\del g}{\del {x^0}}+\Gamma \ast
\nabla\frac{\del g}{\del {x^0}}$, we see that
 $\frac{\del R_{ij}}{\del
{x^0}}\in W^{-1,q}_x$ for all $q \in [1,\infty)$.  Note that $\frac{\del
\lambda}{\del {x^0}}\in W^{2,q}_x, \frac{\del k_{li}}{\del
{x^0}}\in L_x^q,\nabla \lambda\in W^{1,q}_x,
\nabla_i\nabla_j\xi^k\in L_x^q,\nabla \xi \in L_x^q,
\nabla(\frac{\del g}{\del {x^0}})\in
L_x^q,\nabla^2(\frac{\del g}{\del
{x^0}})=\del\nabla(\frac{\del g}{\del {x^0}})+\Gamma \ast
\nabla(\frac{\del g}{\del {x^0}})\in W_x^{-1,q}$.

Now at the boundary $|x|=\eta'$, we have
$$
\frac{\del \xi^k}{\del {x^0}}=\frac{\del^2 x^k}{\del
{t^2}}+\xi(\xi^k)=0
$$
where we used $\xi\mid_{|x|=\eta'}=0$. By applying the $L^p$
estimate, we conclude that $\frac{\del \xi^k}{\del {x^0}}\in
W^{1,q}_x$ for all $q\in [1,\infty)$.
In particular,
\be
\label{3xi}
\frac{\del}{\del {x^0}}\nabla^2\xi^k
= \nabla^2\frac{\del \xi^k}{\del {x^0}}
  + \nabla(\frac{\del g}{\del {x^0}})\ast g^{-1}\ast \nabla \xi
  + g^{-1}\ast \nabla^2(\frac{\del g}{\del {x^0}})\ast \xi\in W^{-1,q}_x.
\ee

In summary, we have proved that the first order (in time)
derivatives of the metric, lapse function, and shift vector have
well-controlled spatial derivatives up to first (or even second)
order: \be \label{1time} \frac{\del g_{ij}}{\del {x^0}}\in
W^{1,q}_x, \quad \frac{\del \lambda}{\del {x^0}}\in W_x^{2,q},
\quad \frac{\del \xi^k}{\del {x^0}}\in W^{1,q}_x \ee for all $q\in
[1,\infty)$.

\

\noindent{\it Step 3.  Second-order time derivative of the metric and lapse function.}

First of all, by differentiating (\ref{evolution}) we find
\be
\label{evolution2}
\aligned
\frac{\del^2 g_{ij}}{\del {x^0}^2}
& = \lambda \, \frac{\del k}{\del {x^0}}+\frac{\del \lambda}{\del x^0} k
     +\frac{\del g}{\del {x^0}}\ast \nabla\xi +\nabla \left( \frac{\del \xi}{\del x^0} \right) \ast g
    +g^{-1}\ast \nabla \left( \frac{\del g}{\del {x^0}}\right) \ast \xi \ast g,
\\
\frac{\del^2k_{ij}}{\del {x^0}^2}&=\nabla^2\frac{\del \lambda}{\del
{x^0}}+\Lcal_{\xi}\frac{\del k}{\del {x^0}}+ \nabla k\ast
\frac{\del \xi}{\del {x^0}}+k\ast \nabla \left( \frac{\del \xi}{\del {x^0}}\right)
\\
& \quad +\lambda \, \Big( k\ast \frac{\del k}{\del {x^0}}\ast
g^{-1}+g^{-2}\ast\frac{\del g}{\del {x^0}}\ast k^2+\frac{\del R{ic}}{\del {x^0}}+k\Big)
\\
& \quad +g^{-1}\ast \nabla\left( \frac{\del g}{\del {x^0}} \right)
\ast \big( \nabla \lambda + k\ast \xi \big) + \left( k^2 \ast g^{-1}+Ric \right) \, \frac{\del \lambda}{\del {x^0}}.
\endaligned
\ee
Recalling that $g_{ij}, \xi^i, \lambda \in W^{2,q}_x$,
$\frac{\del g_{ij}}{\del x^0} \in W^{1,q}_x$,
$\frac{\del \lambda}{\del {x^0}}\in W_x^{2,q}$,
$\frac{\del \xi^k}{\del {x^0}}\in W^{1,q}_x$,
$k\in W^{1,q}_x$,
$\frac{\del k}{\del {x^0}}\in L^q_x$,
$\frac{\del Ric}{\del {x^0}}\in W^{-1,q}_x$,
and combining together with \eqref{metricequation}, \eqref{1time}, and \eqref{spacial},
 we get
the following bounds for the metric and the second fundamental form
\be
\label{t2g}
\frac{\del^2 g_{ij}}{\del {x^0}^2}\in L^q_x,
\quad
\frac{\del^2 k_{ij}}{\del {x^0}^2}\in W^{-1,q}_x
\ee
for all $q\in [1,\infty)$.

To handle the lapse function we note that $\frac{\del \lambda}{\del x^0}= \frac{\del
\lambda}{\del t}+\la\xi, \nabla\lambda\ra$ and
$\frac{\del}{\del x^0}=\lambda N+\xi$, so that
 \be
\aligned
\frac{\del^2\lambda}{\del {x^0}^2}
&  = \frac{\del}{\del x^0}( \lambda N \lambda)+
\frac{\del}{\del x^0}\la\xi,\nabla \lambda\ra\\
\frac{\del}{\del x^0}\la\xi,\nabla
\lambda\ra&=g^{-2}\ast\frac{\del g}{\del x^0}\ast \xi
\ast \nabla \lambda+ g^{-1}\ast \frac{\del \xi}{\del x^0}\ast \nabla\lambda+g^{-1}\ast \xi\ast \nabla\frac{\del
\lambda}{\del x^0}
\\
\frac{\del}{\del x^0}( \lambda N
\lambda)&=\frac{\del^2 \lambda}{\del
t^2}+\xi(\frac{\del \lambda}{\del x^0})-\xi \la
\xi,\nabla\lambda\ra\\
&=\frac{\del^2 \lambda}{\del
t^2}+\la\xi,\nabla\frac{\del \lambda}{\del
x^0}\ra- \la
\nabla_{\xi}\xi,\nabla\lambda\ra-\nabla^2\lambda(\xi,\xi),
\endaligned
\ee
and
\be
\aligned
 \nabla \frac{\del^2\lambda}{\del {x^0}^2}
=
& \nabla\frac{\del^2 \lambda}{\del t^2}
\\
&
 +g^{-2}\ast\nabla \frac{\del g}{\del x^0}\ast \xi \ast
\nabla \lambda+g^{-2}\ast\frac{\del g}{\del x^0}\ast
\nabla\xi \ast \nabla \lambda+g^{-2}\ast\frac{\del g}{\del x^0}\ast \xi \ast \nabla^2 \lambda\\
& + g^{-1}\ast \frac{\del \xi}{\del x^0}\ast
\nabla^2\lambda+ g^{-1}\ast \nabla\frac{\del \xi}{\del
x^0}\ast \nabla\lambda+g^{-2}\ast\nabla^3 \lambda \ast \xi^2
\\
& + g^{-2}\ast\nabla^2\xi\ast\xi\ast\nabla
\lambda+g^{-2}\ast\nabla\xi\ast\nabla \xi\ast\nabla
\lambda+g^{-2}\ast\nabla\xi\ast\xi\ast\nabla^2 \lambda
\\
& + g^{-1}\ast \nabla\xi\ast \nabla\frac{\del \lambda}{\del x^0}+g^{-1}\ast \xi\ast \nabla^2\frac{\del \lambda}{\del x^0}.
\endaligned\ee

Hence, combining together \eqref{xis2c}, \eqref{nbgji1}, and \eqref{1time}, we
arrive at the following control of the lapse function
$$
\int_{|x| \leq \eta'} \left| \nabla \frac{\del^2\lambda}{\del {x^0}^2} \right|^q
     + \left| \frac{\del^2\lambda}{\del {x^0}^2} \right|^q
\leq C_q
$$
for all $q\in [1,\infty)$.

\

\noindent{\it Step 4.  Second-order time derivative of the lapse function.}

It remains to derive the second-order time estimate for the shift function.
By differentiating \eqref{xiequationtime0} in time, we have
\be
\label{xiequationtime2}
\Delta \frac{\del^2 \xi^k}{\del {x^0}^2}
= B_1 +B_2+B_3,
\ee
with
$$
\aligned
B_1 := & -g^{ki}R_{ij}\frac{\del^2 \xi^j}{\del {x^0}^2}+Ric \ast
\Big( \frac{\del \xi}{\del x^0}\ast \frac{\del g}{\del
x^0}\ast g^{-2}+g^{-3}\ast(\frac{\del g}{\del x^0})^2\ast \xi \Big)
\\
&\  + g^{-2}\ast k \ast \nabla\frac{\del^2 \lambda}{\del {x^0}^2}
    + \Big( g^{-2}\ast\frac{\del k}{\del x^0}
    + g^{-3} \ast {\del g \over \del x^0}\ast k \Big)
\ast\nabla\frac{\del \lambda}{\del x^0}\\
& + \left( g^{-2}\ast\frac{\del^2 k}{\del {x^0}^2}
+g^{-3}\ast\frac{\del g}{\del x^0}\ast\frac{\del k}{\del x^0}
+g^{-4}\ast(\frac{\del g}{\del x^0})^2\ast k \right) \ast\nabla \lambda,
\endaligned
$$
$$
\aligned
B_2 := &
 +\frac{\del^2 Ric}{\del {x^0}^2}\ast
g^{-1}\ast\xi +\frac{\del Ric}{\del x^0}\ast \left( \frac{\del \xi}{\del
x^0}\ast g^{-1}+g^{-2}\ast \frac{\del g}{\del x^0}\ast \xi \right)
\\
& +\frac{\del^2 g}{\del {x^0}^2}\ast \left( g^{-2}\ast Ric\ast\xi+
k\ast g^{-3}\ast \nabla \lambda +g^{-2}\ast
\nabla^2\xi\right)
\\
& +(\frac{\del }{\del x^0}\nabla^2\xi)\ast g^{-2}\ast \frac{\del
g}{\del x^0}+\nabla^2\xi\ast g^{-3}\ast (\frac{\del
g}{\del x^0})^2,
\endaligned
$$
$$
\aligned
B_3 :=
& \left( (\nabla\frac{\del g}{\del x^0})^2+\frac{\del g}{\del
x^0}\ast \nabla^2\frac{\del g}{\del
x^0}+\nabla^2\frac{\del^2g}{\del {x^0}^2}\ast g \right) \ast \xi
\ast g^{-3}
\\
&+g^{-2}\ast \nabla^2(\frac{\del g}{\del x^0})\ast \frac{\del
\xi}{\del x^0} +\nabla(\frac{\del g}{\del x^0})\ast g^{-2}\ast
\nabla \frac{\del\xi}{\del x^0}
\\
& +\nabla(\frac{\del^2 g}{\del {x^0}^2})\ast g^{-2}\ast \nabla
\xi +\nabla(\frac{\del g}{\del x^0})\ast g^{-3} \ast \frac{\del
g}{\del x^0} \ast\nabla \xi.
\endaligned
$$
Note that  we already have $\del_{x,t}\gbf$,
$\del_{x,t}^2\gbf\in L_x^q$ for all $q\in [1,\infty)$, {\sl except} that we do not control
$\frac{\del^2\xi}{\del {x^0}^2}$ yet.
Therefore, we can write
\be
\label{xiequationtime2-prime}
\aligned
\Delta \frac{\del^2 \xi^k}{\del {x^0}^2}
= & -g^{ki}R_{ij}\frac{\del^2 \xi^j}{\del {x^0}^2}+g^{-2}\ast\frac{\del^2 k}{\del {x^0}^2}\ast\nabla \lambda
\\
& +\frac{\del^2 Ric}{\del {x^0}^2}\ast g^{-1}\ast\xi +\frac{\del Ric}{\del x^0}\ast
\Big( \frac{\del \xi}{\del x^0}\ast
g^{-1}+g^{-2}\ast \frac{\del g}{\del x^0}\ast \xi \Big)
\\
& +(\frac{\del }{\del x^0}\nabla^2\xi)\ast g^{-2}\ast \frac{\del
g}{\del x^0}+g^{-2}\ast \nabla^2(\frac{\del g}{\del x^0})\ast
\frac{\del \xi}{\del x^0}
\\
& +g^{-3}\ast g\ast\nabla^2(\frac{\del^2 g}{\del
{x^0}^2})\ast \xi+\nabla(\frac{\del^2 g}{\del {x^0}^2})\ast
g^{-2}\ast \nabla \xi \qquad \text{ mod. } L^q_x.
\endaligned
 \ee

In view of
\eqref{spacial}, \eqref{1time}, \eqref{3xi}, and \eqref{t2g} we have
$$
\aligned
& \frac{\del Ric}{\del x^0},\frac{\del^2k}{\del {x^0}^2},
\nabla\frac{\del^2 \lambda }{\del {x^0}^2},\frac{\del}{\del x^0}\nabla^2\xi,
 \nabla(\frac{\del^2 g}{\del {x^0}^2}),
\nabla^2(\frac{\del g}{\del {x^0}})\in W^{-1,q}_x,
\\
& \nabla^2\frac{\del^2g}{\del {x^0}^2}\in W^{-2,q}_x,
\endaligned
$$
and
\begin{equation*}
\begin{split}
\frac{\del^2 Ric}{\del {x^0}^2}&=g^{-1}\ast \, \Big( \nabla
^2\frac{\del^2g}{\del {x^0}^2} + Rm\ast \frac{\del^2
g}{\del {x^0}^2}
\\
& \quad +\nabla^2(\frac{\del g}{\del
x^0})\ast(\frac{\del g} {\del x^0})\ast g^{-1}+\nabla(\frac{\del
g}{\del x^0})\ast \nabla(\frac{\del g}
{\del x^0})\ast g^{-1}\Big)
\\
&= \nabla ^2(g^{-1}\ast\frac{\del^2g}{\del {x^0}^2}) +\nabla
\Big( \nabla\frac{\del g}{\del
x^0}\ast\frac{\del g} {\del x^0}\ast g^{-2}\Big)\ \ \  \text{mod.} \ \ L^q_x\\
\end{split}
 \end{equation*}
 for all $q\in [1,\infty)$. Consequently, we have
\be
\label{xiequationtime2-third}
\aligned
\Delta \frac{\del^2 \xi^k}{\del {x^0}^2}+g^{ki}R_{ij}\frac{\del^2 \xi^j}{\del
{x^0}^2}&= \nabla_i\nabla_j ({f}^{kij}_m \xi^m)+\nabla_i f^{ki}+f^k\\
&=\del_i\del_j (F^{kij}_m)\xi^m+\del_i F^{ki}+F^k,
\endaligned
\ee
where for fixed $k$, $f^{kij}_m$  etc. are tensors,
and $\nabla$ and $\del$ are covariant derivatives and partial derivatives in the coordinates $x^i$,
respectively, with moreover
$$
\int_{|x|\leq \eta'}
\Big( |f^{kij}|^q + |f^{ki}|^q + |f^k|^q + |F^{kij}_m|^q + |F^{ki}|^q + |F^k|^q\Big)
\leq C_q.
$$
For the second equality, we used $\xi \del^2g\in L^q$.
Now, we will use the $L^p$ regularity estimates in the following manner.

Since the coefficients of the elliptic operator $g^{ab} \del_a \del_b $ belong to $C^\gamma$
on the closed ball $\big\{ |x|\leq \eta' \big\}$, we can solve the equation
$g^{ab}\del_a\del_b u^{kij} = F^{kij}$ on $|x|< \eta'$ with the trivial boundary condition
$u^{kij}\mid_{|x|=\eta'}=0$. We then apply the $L^p$ regularity estimate and obtain
$$
\int_{|x|\leq \eta'} \Big(
{|u^{kij}|^q+|\del u^{kij}|^q +|\del^2 u^{kij}|^q} \Big)
\leq C_q.
$$
Next, we observe that
$$
\aligned
\del_i\del_j (F^{kij}_m)\xi^m
=
&  g^{ab}\del_a\del_b(\del_i\del_j u^{kij}_m\xi^m)+\xi \del^2 g \ast
\del^2 u^{kij}_m
\\
&
+\del(\del g \ast \del^2 u^{kij}_m \ast \xi
+\del^2 u^{kij}_m \ast \del \xi) +\del^2
u^{kij}_m\ast \del^2 \xi \ast g^{-1}
\endaligned
$$
and so, for some new terms $F^{ki}$, we obtain
\be
\label{xiequationtime3}
\aligned
\Delta \frac{\del^2 \xi^k}{\del {x^0}^2}+g^{ki}R_{ij}\frac{\del^2 \xi^j}{\del {x^0}^2}
& = g^{ab}\del_a\del_b( u^k_m\xi^m)+\del_i F^{ki}+F^k
\endaligned
\ee
for some
\be
\label{qguai}
\int_{|x|\leq \eta'}|u^k_m|^q\leq C_q.
\ee
Since
$$
\aligned
\Delta \big( u^k_m\xi^m \big)
& = g^{ab}\del_a\del_b( u^k_m\xi^m)+ \Gamma \ast \del (u\ast \xi)+ \del^2 g \ast u\ast \xi
\\
& = g^{ab}\del_a\del_b( u^k_m\xi^m) + \del (\Gamma \ast u\ast \xi)
   + \del^2 g \ast u \ast \xi,
\endaligned
$$
and $\xi \del^2 g \in L^q$ for all $q$, by modifying $F^{ki}$ and $F^k$,
 we can show
\be
\label{xiequationtime4}
\aligned
\Delta \left( \frac{\del^2 \xi^k}{\del {x^0}^2} - u^k_m\xi^m \right)
     + g^{ki}R_{ij} \left( \frac{\del^2 \xi^j}{\del {x^0}^2}- u^j_m\xi^m \right)
& = \del_i F^{ki} + F^k,
\endaligned
\ee
where the notation $\Delta$ stands here for the covariant Laplacian of a vector field.

It is not hard to see $\frac{\del^2 \xi^k}{\del
{x^0}^2}\mid_{|x|=\eta'}=0$. Let $v^k=\frac{\del ^2
\xi^k}{\del {x^0}^2}-u^k_m\xi^m$. Integrating
$\Delta |v^k|^{2q+2}$ with the induced (intrinsic) volume form, using
$ v^k\mid_{|x|=\eta'}=0$ and \eqref{xiequationtime4}, and finally applying
H\"older inequality, we find
\be
\label{xiequationtime5}
\aligned
\int_{|x|\leq \eta'}|v^k|^{2q}|\nabla v^k|^2
\leq
C_q \, \int_{|x|\leq \eta'}|v^k|^{2q+2}+C_q
\endaligned
\ee for all $q\in [1,\infty)$. This implies, in particular, $\int_{|x|\leq
\eta'}|v^k|^2q\leq C$ by Sobolev inequalities.
Combining this result with \eqref{qguai}, we arrive at the estimate for the shift vector
$$
\int_{|x|\leq \eta'} \left| \frac{\del^2 \xi^k}{\del {x^0}^2} \right|^q. 
\leq C_q
$$

In summary, we have obtain the following uniform control of the second-order time derivatives
of the metric, lapse function, and shift vector:
\be
\label{2time}
\frac{\del^2 g_{ij}}{\del {x^0}^2}\in L^q_x,
\quad
\frac{\del^2 \lambda}{\del {x^0}^2}\in L_x^q,
\quad
\frac{\del^2 \xi^k}{\del {x^0}^2}\in L^q_x
\ee
for all $q \in [1,\infty)$.
By combining (\ref{spacial}) with (\ref{1time}) and (\ref{2time}),
the proof of Theorem \ref{main} is now completed.


\section*{Acknowledgements}

The authors thank Lars Andersson for providing them with bibliographical informations. The first author (BLC) 
was partially supported by Sun Yat-Sen University (Guangzhou) through a grant ``New Century Excellent
Talents'' (NCET-05-0717). The second author (PLF) was partially
supported by the Centre National de la Recherche Scientifique
(CNRS) and the Agence Nationale de la Recherche (ANR) through the
grant 06-2-134423 entitled {\sl ``Mathematical Methods in General Relativity''} (MATH-GR).


\end{document}